\documentclass[10pt,letterpaper]{article}
\usepackage{amsmath}
\usepackage{amssymb}
\usepackage{fullpage}
\usepackage{amsfonts}
\usepackage{graphicx}
\usepackage{float}
\usepackage{wrapfig}
\usepackage{verbatim}
\usepackage{color}
\usepackage{amsthm}
\usepackage{caption}
\usepackage{subcaption}
\usepackage{hyperref}
\usepackage[toc,page]{appendix}
\usepackage{enumerate}
\usepackage{mathtools}
\usepackage{mathrsfs}
\usepackage{epstopdf}
\usepackage{tikz}
\usepackage{sidecap}

\usepackage{xparse}
\ExplSyntaxOn
\NewDocumentCommand{\mref}{m}{\quinn_mref:n {#1}}
\seq_new:N \l_quinn_mref_seq
\cs_new:Npn \quinn_mref:n #1
{
	\seq_set_split:Nnn \l_quinn_mref_seq { , } { #1 }
	\seq_pop_right:NN \l_quinn_mref_seq \l_tmpa_tl
	( 
	\seq_map_inline:Nn \l_quinn_mref_seq
	{ \ref{##1},\nobreakspace } 
	\exp_args:NV \ref \l_tmpa_tl 
	) 
}
\ExplSyntaxOff

\usepackage[export]{adjustbox}
\usepackage{datetime}
\usepackage{calc}
\mathtoolsset{showonlyrefs=true}
\usetikzlibrary{matrix,chains,positioning,decorations.pathreplacing,arrows}
\definecolor {processblue}{cmyk}{0.96,0,0,0}

\newtheorem{theorem}{Theorem}

\newtheorem{corollary}{Corollary}[theorem]
\newtheorem{lemma}[theorem]{Lemma}
\newtheorem{prop}{Proposition}
\theoremstyle{definition}

 {
      \theoremstyle{plain}
      
  }

\newcommand{\equald}{\,{\buildrel d \over =}\,}

\newcommand{\tr}{\textcolor{red}}

\newcommand{\pr}{\mathbb{P}}
\newcommand{\ep}{\mathbb{E}}

{\tiny {\tiny {\normalsize {\normalsize }}}}

\newcommand{\rmd}{\mathrm{D}}
\newcommand{\rmdd}{\mathrm{DD}}
\newcommand{\rms}{\mathrm{S}}
\newcommand{\rmt}{\mathrm{Tot}}

\newcommand{\mic}{\textcolor{black}}

\newenvironment{redsect}{\par\color{black}}{\par}

\def \birthone{\alpha}
\def \deathone{\beta}
\def \growthone{\lambda}
\def \pl{l}

\sidecaptionvpos{figure}{t}

\def \survone{S_1}
\def \rvec{\Phi}
\def \walk{\omega}

\title{Competing evolutionary paths in growing populations\\ with applications to multidrug resistance  }
\usepackage{authblk}
\date{}                     
\author[1]{Michael D. Nicholson}
\author[2]{Tibor Antal}
\affil[1]{ School of Physics and Astronomy, University of Edinburgh}
\affil[2]{School of Mathematics, University of Edinburgh}
\begin{document}


	\maketitle
	
	\section*{Abstract}
	Investigating the emergence of a particular cell type is a recurring theme in models of growing cellular populations. The evolution of resistance to therapy is a classic example. Common questions are: when does the cell type first occur, and via which sequence of steps is it most likely to emerge?  For growing populations, these questions can be formulated in a general framework of branching processes spreading through a graph from a root to a target vertex. Cells have a particular fitness value on each vertex and can transition along edges at specific rates. Vertices represents cell states, say \mic{genotypes }or physical locations, while possible transitions are acquiring a mutation or cell migration. We focus on the setting where cells at the root vertex have the highest fitness and transition rates are small. Simple formulas are derived for the time to reach the target vertex and for the probability that it is reached along a given path in the graph. We demonstrate our results on \mic{several scenarios relevant to the emergence of drug resistance}, including: the orderings of resistance-conferring mutations in bacteria and the impact of imperfect drug penetration in cancer. 
	
	\tableofcontents

\section{Introduction}\label{sec_intro}
	\footnotetext[1]{mdnicholson5@gmail.com}
		\footnotetext[2]{Tibor.Antal@ed.ac.uk}
%

The timing and manner in which a particular phenotype arises in a population is a central question of theoretical biology \cite{Berestycki:2016,Beerenwinkel:2007,Weinreich:2006,Weissman:2009,Durrett:2009,Gillespie:1983, Zagorski:2016, Gokhale:2009,DeVisser:2014,Tadrowski:2018,Hermsen:2010,Nowak:2002,Nichol:2015}. A typical scenario is to consider an initially monomorphic, wild type population, composed of cells that can acquire mutations, for example single site substitutions on the genome. The phenotype of interest comes to exist after a cell has accrued a specific set of mutations. The interpretations of this event are application dependent, but examples are the genesis of cancer instigated by mutations in a pair of tumour suppressor genes, or the emergence of multidrug resistance via alterations to the genes coding for the target proteins. Regardless of context, the questions of when, and how, the phenotype emerges are of significant interest.

\mic{It is commonly assumed that the population under consideration is a fixed size}. However, with the aim of characterising disease progression, an increasing body of research has been developed to examine the evolutionary dynamics of a growing population. These studies have provided insights on a range of applications, including; cancer genetics \cite{Williams:2016,Bozic:2016,Cheek:2018,Gao:2016}, metastasis formation \cite{Haeno:2007,Haeno:2012,Nicholson:2016,Dewanji:2011,Komarova:2005,Reiter:2018}, drug resistance \cite{Bozic:2013,Kessler:2014,Komarova:2005,Ford:2013}, phylogenetics \cite{Wilkinson:2009}, and the impact of poor drug penetration \cite{Greulich:2012,Hermsen:2012,Fu:2015,Moreno-Gamez:2015}.



Here we continue in the same vein by considering a stochastically growing cellular population, where cells can transition in such a fashion \mic{so} as to alter their, and their offsprings, reproductive capabilities. Such a transition might be due to the acquisition of a (epi)genetic alteration or migration into a new environment. As before, suppose we have a cellular state of interest, for example; a given genotype, a spatial location, or a combination of both. Will this state ever be reached? If it is, when is it reached? And by which sequence of intermediate states? 

To make the discussion clear, let us consider an example application: the emergence of multidrug resistant bacteria.  Suppose an infection begins with a single pathogenic bacterium which is sensitive to two antibiotic therapies, drug $A$ and drug $B$. In the absence of either drug, the initial bacterium initiates a growing colony. Within the colony, when a cell replicates one of the daughter cells can acquire a mutation yielding resistance to either of the drugs. Here our questions are: how long does it take for multidrug resistance to emerge? En route, is resistance to drug $A$ or drug $B$ more likely to arise first? \mic{An ability to answer such questions is key to understanding pre-existing resistance, a common cause of therapy failure in various settings \cite{Bozic:2017}} The scenario is illustrated in Fig.\ \ref{fig_example}. There each vertex represents a cellular type, in this case its resistance profile. The edges represent cell transitions via mutation upon replication.

\mic{
In this article we focus on the setting where the intermediate states have reduced reproductive ability (fitness) relative to the initial cells. This is primarily motivated by the commonly observed cost of resistance \cite{Andersson:2010,Hughes:2015,Enriquez-Navas:2016}, whereby cellular populations that are resistant to a given drug grow more slowly than their sensitive counterparts. A second scenario where we expect a reduction in fitness is drug sensitive cells becoming exposed to toxins, which increases the cell's rate of dying (cytotoxic) or decreases the rate of replication (cytostatic).} As in many biological applications the relevant transition rates are small (\mic{some representative examples are}: measurements for the point mutation rate per cell division of $10^{-9}$ in cancer \cite{Tomasetti:2013} and $10^{-10}$ for bacteria \cite{Lee:2012}, while the dissemination rate of pancreatic cancer cells from the primary tumour was estimated as $10^{-7}$ per cell division in \cite{Haeno:2012}), we concentrate our efforts on the regime of low transition rates. Our main contribution is to provide simple, intuitive formulas that answer the questions posed in the small transition rate limit. These formulas show explicitly the contribution of the model parameters, e.g.\ transition rates and fitness reductions, which allow them to be probed for biological insight. This provides relationships which would be difficult to deduce from simulation alone.

We now move to detail the general framework we study, which will be seen to encompass models introduced in previous works as special cases \cite{Fu:2015,Komarova:2005,Ford:2013,Colijn:2011}. Our main results concerning when and how a particular cell type emerges are then presented, and we follow this by demonstrating our method on several applications.

\section{Results}\label{sec_results}
\subsection{General framework}\label{sec_model}
\begin{figure}[t!]
	\begin{subfigure}[t]{.45\textwidth}
		\centering
		\resizebox{5cm}{5cm}{
			\begin {tikzpicture}[-latex ,auto ,node distance =5.3 cm and 2cm ,on grid ,
			very thick ,
			state/.style ={ circle ,top color =white , bottom color = processblue!20 ,
				draw,processblue , text=blue ,minimum size =2.7cm,font={\Large}}]
			\node[state,text width=2cm,align=center,label={[label distance=0.12cm]90:\huge root}] (1)  at (0,0){Sensitive bacteria};
			\node[state,text width=2cm,align=center] (2)  [below right of = 1] {Resistant to drug $B$};
			\node[state,text width=2cm,align=center] (3)  [below left of = 1] {Resistant to drug $A$};
			\node[state,text width=2cm,align=center,label={[label distance=-4cm]90: \huge target}] (4)  [below right of = 3] {Resistant to drugs $A$ \& $B$};
			
			\path (1) edge[line width=1.8pt] node[above =0.15 cm,rotate=-45,font={\Large}] {Mutation} (2);
			\path (1) edge[line width=1.8pt] node[right =0.3 cm,above=.13cm,rotate=45,,font={\Large}] {Mutation} (3);
			\path (2) edge[line width=1.8pt] node[right =0.3 cm,above=.13cm,rotate=45,,font={\Large}]{Mutation} (4);
			\path (3) edge[line width=1.8pt] node[above =0.15 cm,rotate=-45,font={\Large}] {Mutation} (4);

		\end{tikzpicture}
	}
	\caption{}
	\label{fig_example}
\end{subfigure}
\,
\begin{subfigure}[t]{.45\textwidth}
	\centering
	\begin{tikzpicture}[-latex ,auto ,node distance =1.2 cm and 1.8cm ,on grid ,
	semithick ,
	state/.style ={ circle ,top color =white , bottom color = processblue!20 ,
		draw,processblue , text=blue , minimum width =.5 cm},scale=0.6, every node/.style={scale=0.6}]
	\tikzstyle{every node}=[font=\small]
	\node[state,label={[label distance=0.10cm]90:root}] (1)  at (0,3){$1$};
	\node[state] (2) [below left = of 1] {$2$};
	\node[state] (3) [below = of 1] {$3$};
	\node[state] (4) [below right = of 1] {$4$};
	\node[state] (5) [below left = of 3] {$5$};
	\node[state] (6) [below = of 3] {$6$};
	\node[state] (7) [below right = of 3] {$7$};
	\node[state,label={[label distance=-1.45cm]90:target}] (8) [below = of 6] {$8$};
	
	\path (1) edge[blue] node[above =0.15 cm] {} (2);
	\path (1) edge node[above =0.15 cm] {} (3);
	\path (1) edge[red] node[above =0.15 cm] {} (4);
	
	\path (3) edge node[above =0.15 cm] {} (2);
	
	\path (2) edge[blue] node[above =0.15 cm] {} (5);
	\path (2) edge node[above =0.15 cm] {} (7);
	\path (3) edge node[above =0.15 cm] {} (6);
	\path (4) edge[red] node[above =0.15 cm] {} (8);
	\path (5) edge[blue] node[above =0.15 cm] {} (8);
	\path (6) edge[black] node[above =0.15 cm] {} (8);
	\path (7) edge node[above =0.15 cm] {} (8);
	
	\end{tikzpicture}
	\caption{}
	\label{fig_graph}
\end{subfigure}
\caption{(a) Motivating example: consider a population of growing drug sensitive bacteria. Each vertex represents a cell type and cells can mutate to acquire resistance. Starting from one sensitive bacterium that can replicate, die or mutate, how long does it take to traverse the graph until multidrug resistance emerges? Is resistance to drug $A$ then $B$ more likely than the converse? (b) For the general case, the growing population is composed of different cell types that can be associated to vertices on a graph. Cells at vertex $x$ can transition to vertex $y$ if the edge $(x,y)$ exists. Starting with cells at the root vertex, how long until the target is populated? Is the blue or red path more likely to initiate the target population? Here $N=8$.}
\end{figure}
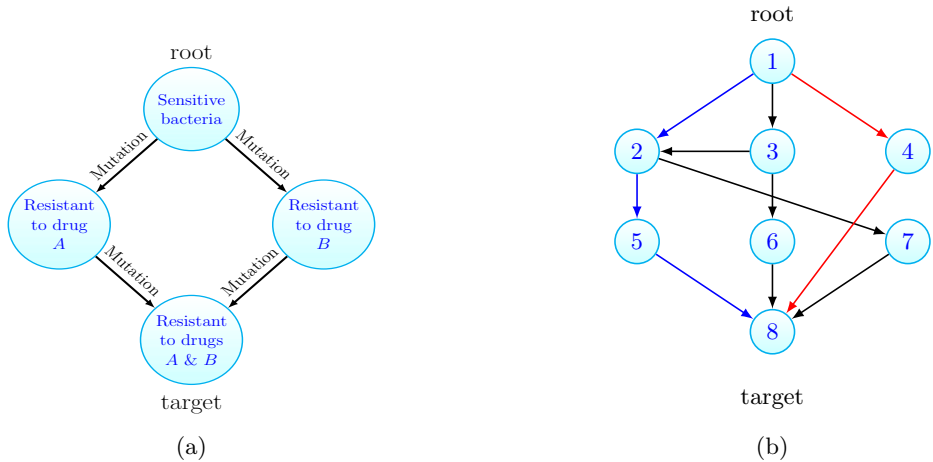


Our interest will always be in the emergence of a particular cell type in a growing population which is modelled as a specific form of a multitype branching process \cite{Athreya:2004}.  It is convenient to picture the population on a finite, simple, directed graph $G=(V,E)$, see Fig.\ \ref{fig_graph}, where each vertex of the graph represents a cell type or state. The number of types in the system is denoted $N$ and so we let $V=\{1,2,\ldots,N\}$. Thus $E$ is a subset \mic{of} the set of ordered pairs $\{(i,j):i,j\in V,i\neq j\}$. For any given cell type there is an associated vertex in $V$. Henceforth we will refer to cells of type $x$ as cells residing at vertex $x$. We will be concerned with the timing, and fashion, that vertex $N$, which we refer to as the target, is populated, when we initiate with cells at vertex 1, which we denote the root. In the example of multidrug resistance, illustrated in Fig.\ \ref{fig_example}, the cells at the root are sensitive bacteria while the target population is resistant to both drugs.

Let $(x)$ represent a cell at vertex $x$ and $\varnothing$ symbolise a dead cell. Then, with all cells behaving independently, our cell level dynamics can be graphically represented as:
\begin{align*}\label{model_illust}
(x) \rightarrow 
\begin{cases}
(x),(x) \quad &\mbox{at rate } \alpha(x)
\\
\varnothing \quad &\mbox{at rate } \beta(x)
\\
(x),(y) \quad &\mbox{at rate } \nu(x,y) \mbox{ if }  (x,y)\in E.
\end{cases}
\end{align*}
That is: cells \mic{divide} at rate $\alpha(x)$, die at rate $\beta(x)$, and transition to a cell at vertex $y$ at rate $\nu(x,y)$ if the edge $(x,y)$ exists. We will denote the fitness (growth rate) of cells at vertex $x$ by $\lambda(x)=\alpha(x)-\beta(x)$, the difference of birth and death rates of vertex $x$ cells. The parameters associated with the vertex 1 population feature prominently and so for convenience we let $\alpha=\alpha(1),\,\beta=\beta(1)$ and $\lambda=\lambda(1)$. As mentioned above, in this article we focus on the setting where the vertex 1 population has the largest fitness, which is also positive.  That is, we assume that $\growthone>0$ and for $2\leq x \leq N-1$, $\lambda(x) < \growthone$. We do not specify the fitness of the target population (cells at vertex $N$).

 \mic{
	A common variant when modeling transitions (or mutations) is to specify that cells of type $x$ divide at rate $\alpha'(x)$, and then with probability $\nu'(x,y)$ a transition occurs to vertex $y$. For a fixed value of parameters, this formulation of transitions is equivalent to that given above upon letting $\nu(x,y)=\alpha'(x)\nu'(x,y)$, $\alpha(x)=\alpha'(x)(1-\sum_{y: (x,y)\in E}\nu'(x,y))$. 
}
 Also in the model as currently stated, transitioned cells have deterministically fixed new birth and death rates. However including a finite distribution of fitness effects, when a transitioned cell is assigned random birth and death rates, is actually covered by our model as given thus far. We discuss this point further in Section \ref{sec_cyc}.

At a population level, the number of cells at vertex $x$ at time $t$ will be denoted $Z_{x}(t)$. We shall always assume that the system is initiated with only vertex 1 cells, at a quantity $z$, that is $Z_{x}(0)=z\delta_{x,1}$ with the Kronecker delta function $\delta_{x,y}$. 

We have two primary questions. The first is, having initiated the system with $z$ vertex 1 cells, how long does it take for the population at the target vertex to arise? That is we concern ourself with the distribution of the target hitting time, defined as
\begin{equation}\label{def_time}
T=\inf\{t\geq 0:\, Z_{N}(t)>0\}. 
\end{equation}

Now assuming that the target vertex is populated by a founding cell, we ask from which path through the graph $G$ did this founding cell come? This gives rise to a distribution over the set of paths (or walks) from the root vertex to the target vertex, which we aim to characterise. This second question is more precisely formulated in Section \ref{sec_results_path}. Throughout we will assume that we wait for the first cell of vertex $N$ to arise, however it may be that the population founded by the initial cell goes extinct. If instead one wishes to wait for the first `successful' cell, that is the first cell of at vertex $N$ to exist whose progeny survives, then all the results presented below hold so long as the mapping $\nu(x,N)\mapsto\nu(x,N)\lambda(N)/\alpha(N)$ for all target adjacent vertices $x$, is applied.

Our results are most clearly understood when $G$ is acyclic, which amounts to excluding reverse transitions. Therefore in the following presentation this will be assumed. The setting when this assumption is relaxed will be discussed in Section \ref{sec_cyc}.


Consider any path from the root to the target vertex. Intuitively, if the transition rates encountered along this path are larger, the target population will be seeded from this path quicker. Conversely, the smaller the growth rates along this path are, the slower the target will be reached. 
\mic{It will transpire that this intuition is indeed correct. The key quantities that show how the time to the target, and which path is taken, depend on these competing factors will be seen to be the path weights, which we now introduce.}

%

Let the set of paths between the root and the target be be denoted $\mathcal{P}_{1,N}$. Any $p\in \mathcal{P}_{1,N}$, will be a sequence $p=(p_1,p_2,\ldots,p_\pl,p_{\pl+1})$, where each $1\leq p_i\leq N$, $\pl$ is the path length (the number of edges in $p$), all $p_i$ are distinct \mic{(as presently $G$ is acyclic)}, and $p_1=1,\,p_{\pl+1}=N$. Along the path $p$, let us call the difference between the growth rate at the root and the vertices encountered the fitness costs. Then we define the weight of path $p$ as
\begin{equation}\label{def_pathweight}
w(p)=\nu(p_1,p_2)\prod_{i=2}^{l}\frac{\nu(p_{i},p_{i+1})}{\growthone-\lambda(p_i)},
\end{equation}
that is the product of the transition rates along the path divided by the fitness costs along the path. Throughout the empty product is set to 1. Further, we let the total weight of the target be 
\begin{equation}\label{def_totalweight}
\phi_N = \sum_{p\in \mathcal{P}_{1,N}}w(p).
\end{equation}
We will take the case of a path graph as a running example. A path graph, as illustrated in Fig.\ \ref{fig_lineargraph1}, is the scenario in which the only edges are $(i,i+1)$, for $1\leq i\leq N-1$. There is only one path $p$ between the root and the target and so in this case, $\phi_N=w(p)=\nu(1,2)\prod_{i=2}^{N-1}\frac{\nu(i,i+1)}{\growthone-\lambda(i)}$.

\begin{figure}

\begin{center}

	\begin{tikzpicture}[-latex ,auto ,node distance =3.5 cm and 2cm ,on grid ,
	semithick ,
	state/.style ={ circle ,top color =white , bottom color = processblue!20 ,
		draw,processblue , text=blue , minimum width =.11 cm}]	
	
	\node[state,label={[label distance=0.12cm]90:\large root}] (v1) at (-5,0) {1};
	
	\node[state](v2) at (-2.5,0) {2};
	\node[state] (v3) at (0,0) {3};
	\node[state,label={[label distance=0.12cm]90:\large target}] (v4) at (2.5,0) {4};

	\path (v1) edge node{$\nu(1,2)$} (v2);
	\path  (v2) edge node{$\nu(2,3)$} (v3);
	\path (v3) edge node{$\nu(3,4)$} (v4);

	\end{tikzpicture}
	\captionof{figure}{A path graph of length 3, with $N=4$.}\label{fig_lineargraph1}
\end{center}
\end{figure}
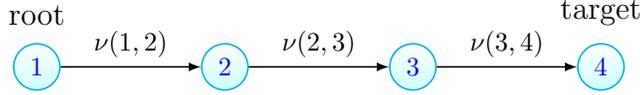

\subsection{Time until target vertex is populated}\label{sec_times}
How long do we wait until the target vertex is populated? We answer this question by characterising the distribution of $T$, defined in \eqref{def_time}, when the target seeding transition rates (the transition rates associated with edges into the target vertex) are small. The key tool required is the long-time population numbers at the initial and intermediate vertices, which is discussed further in Section \ref{sec_methods}. Using this we are able to prove the following theorem regarding the target hitting time, \mic{whose proof can be found in the supplementary material (see Sections \ref{sec_SItimes} and \ref{sec_SIacyclic})}.

\begin{theorem}\label{thm_timesacyclic} 
	As the target seeding transition rates tend to 0, we have 
	\begin{equation}\label{eqn_thmtim}
	\pr(T-\mu>t)\rightarrow\left(\frac{\growthone/\birthone}{1+e^{\lambda t}}+\deathone/\birthone\right)^z
	\end{equation}
	where 
	$$
	\mu=\frac{1}{\growthone}\log\frac{\growthone^2}{\birthone \phi_{N}}.
	$$
\end{theorem}
Heuristically, the time until the target is populated is approximately $\mu$ plus some noise, where the distribution of the noise is given by the right hand side of \eqref{eqn_thmtim}. Note that this distribution depends only on parameters associated with vertex 1 cells, while all other parameters of the system are bundled into $\mu$. For practical purposes the theorem yields the following approximation
	\begin{equation}\label{eqn_timeapprox}
	\pr(T>t)\approx\left(\frac{\growthone/\birthone}{1+e^{\lambda t}\phi_N \alpha/\lambda^2}+\deathone/\birthone\right)^z
	\end{equation}
	\mic{This approximation will be valid if the target seeding transition rates are small. We do not provide estimates on the approximation error, however simulations (for example see Fig.\ \ref{fig_growthtime}b) demonstrate this approximation holds even for only moderately small transition rates.}

Notice that it is possible that the target is never populated ($T=\infty$). In the supplementary material (Proposition \ref{prop_setdif}) we show that if the vertex 1 population survives forever then the target hitting time is finite with probability one. Furthermore in many relevant cases, namely a large initial population, low death rate, or small transition rates leaving vertex 1, that the target is eventually populated is essentially equivalent to the vertex 1 population surviving. \mic{The probability that the vertex 1 population survives is obtained in the supplementary material (Eq. \eqref{eqn_survoneprob}), and this quantity thus} gives an approximation for the probability that the target population will ever arise, which is
\begin{equation}\label{eqn_proboccurs}
\pr(T<\infty)\approx 1-(\deathone/\birthone)^z.
\end{equation}
This also explains the $\beta/\alpha$ term in \eqref{eqn_timeapprox}, which arises due to lineages stemming from the initial vertex 1 cells going extinct. The corresponding approximate distribution for the target hitting time, when we condition that the vertex 1 population survives is given in the supplementary material (Eq.\ \eqref{eqn_timeapprox_cond}). \mic{Differentiating the conditional hitting time distribution yields the density which is compared with simulations in Fig.\ \ref{fig_growthtime}c.}

Let us now discuss the time centring term $\mu$, as given in Theorem \ref{thm_timesacyclic}. In the case of $\beta=0,\,z=1$, it can be easily seen (recall $\growthone=\birthone-\deathone$) that the median of our approximate distribution \eqref{eqn_timeapprox}, is exactly $\mu$. More generally, if we let $t_{1/2}$ be the median time to hit the target, then if the vertex 1 population survives, we have
\begin{equation}\label{eqn_medtime}
t_{1/2}\sim \mu-h(z).
\end{equation}
as the target seeding transition rates tend to 0.
The shift $h(z)$, whose precise form can be found in the supplementary material (Corollary \eqref{cor_medpath}), has the behaviour
\begin{align}\label{eqn_hasym}
h(1)=0,\quad h(z)&\sim \frac{1}{\growthone}\log\frac{z \growthone}{\birthone },\quad z \rightarrow \infty.
\end{align}
The shift exists as initiating the system with a larger number of cells at vertex 1, leads to the target being reached faster. In terms of notation $f\sim g$ means that $f/g\rightarrow 1$, when the limit under consideration is taken.

We now return to the running example of the path graph setting which was introduced at the end of Section \ref{sec_model}. With $z=1$, \eqref{eqn_medtime} yields that the median time for the target population to appear is
\begin{equation}\label{eqn_medpath}
t_{1/2}\approx \mu= \frac{1}{\growthone}\log\frac{\growthone^2 }{\birthone  \nu(1,2)}+\frac{1}{\growthone}\sum_{i=2}^{N-1}\log\frac{\growthone-\lambda(i)}{\nu(i,i+1)},
\end{equation}
for small $\nu(N-1,N)$. The first \mic{summand} on the right hand side comes from waiting for the first transition from the vertex 1 population, while the second is due transitions between the remaining vertices. The first \mic{summand} is distinct as the vertex 1 growth is the main cause of stochasticity, as discussed in Section \ref{sec_methods}. The population growth and target hitting time is illustrated for the path graph case in Fig.\ \ref{fig_growthtime} .

\begin{figure}
	\centering
	\includegraphics[width=.75\linewidth,height=12cm]{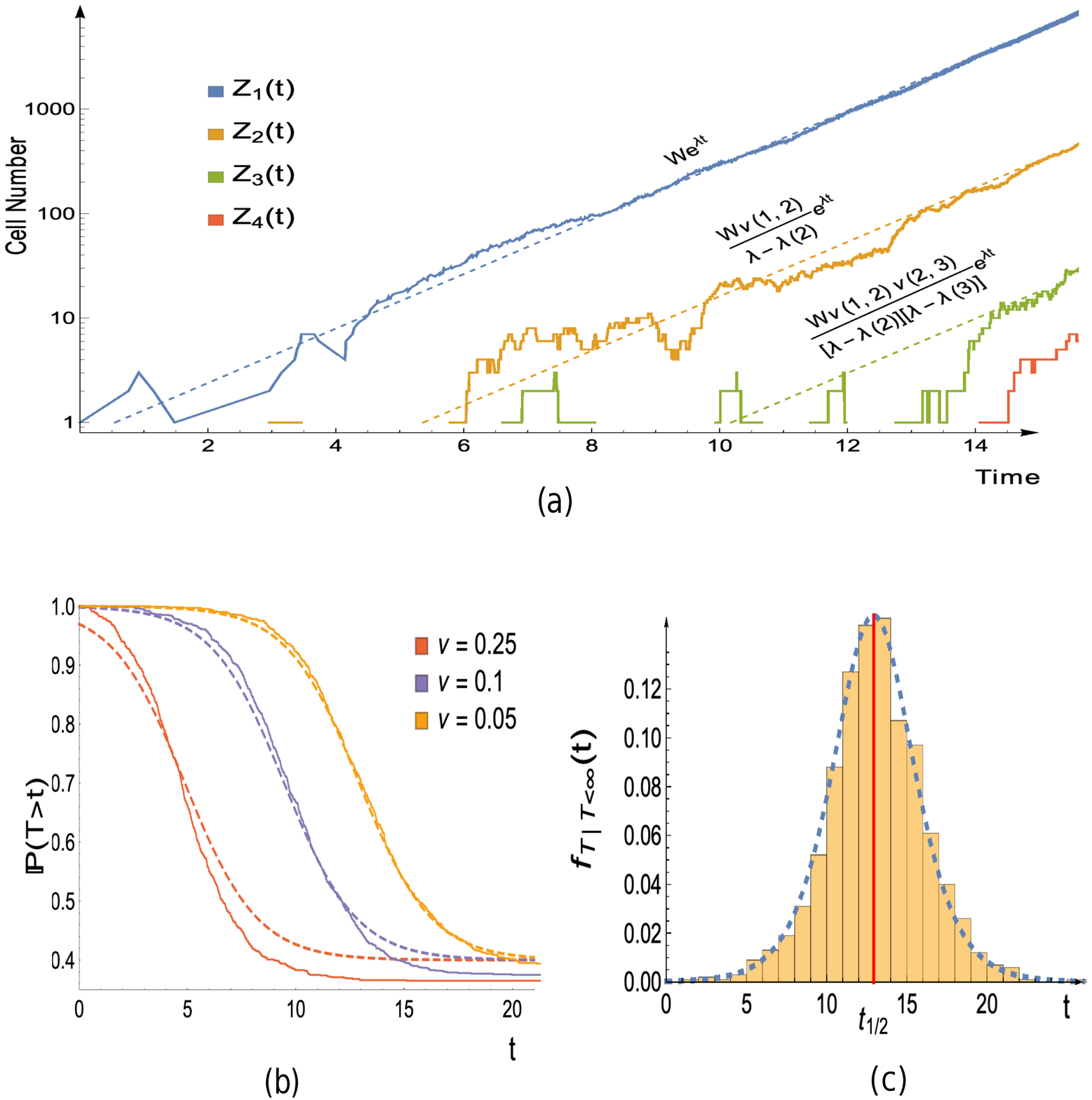}
			\caption{\protect\rule{0ex}{0ex}\mic{ Simulation versus theory in the case of a path graph with 4 vertices, as in Fig.\ \ref{fig_lineargraph1}.  (a) Joined lines show a single realisation of the population growth obtained via simulation. Dashed lines are the asymptotic growth of the populations, which follows the functions given immediately above the dashed lines. These functions are a consequence of Theorem \ref{thm_Zk1}, which can be found in Section \ref{sec_methods}. $W$ is a random variable whose value will differ between simulations. Its value here is estimated by the realised value of $e^{-\growthone t_f}Z_{1}(t_f)$, with $t_f=14.5$. (b) The distribution of the hitting time for vertex 4: Joined lines show the empirical distribution obtained from 1000 realisations, dashed lines display the theoretical approximation \eqref{eqn_timeapprox}.  The error in the approximation decreases as $\nu\rightarrow 0$, but still shows good agreement for moderate values of $\nu$. (c) Probability density function for hitting time of vertex 4, conditioned that it is finite:	bars display normalised histogram constructed from 1000 realisations, while dashed line shows approximation which follows from differentiating the conditional variant of \eqref{eqn_timeapprox} (precisely given in \eqref{eqn_pdfcond} of the supplementary material).  The approximate median time for the type 4 population to arise, as given \eqref{eqn_medpath},  is indicated by the red vertical line. Parameters: $\alpha=\alpha(2)=\alpha(3)=\alpha(4)=1,\beta=0.4,\beta(2)=\beta(3)=1.3,\beta(4)=0.6, z=1$; for (a) and (c), $ \nu(1,2)=\nu(2,3)=\nu(3,4)=0.05$; for (b), legend displays value of $\nu=\nu(1,2)=\nu(2,3)=\nu(3,4)$ used.}}
	\label{fig_growthtime}
\end{figure}

\subsection{Path distribution}\label{sec_results_path}
We now move to our second question: which path leads to the target population arising?  Our naive expectation is that paths with larger weights will be more likely. This simple conjecture turns out to be true. To show this we first introduce some notation.

To any existing cell, say at vertex $x$, we may define the cell's vertex lineage. This tracks the vertices of cells in the ancestral lineage of the cell under consideration and is a sequence of the form $(1,\ldots, x)$. For example a cell with vertex lineage $(1,3,4)$ is at vertex 4, and descended from a vertex 3 cell, who in turn descended from a vertex 1 cell. Any vertex lineage is a path in $G$ and we denote the number of cells with vertex lineage $q$ at time $t$ as $X(q,t)$. \mic{ Note that while there may be multiple ancestors at the same vertex in a given cell's ancestral lineage, for example several generations of vertex 3 cells before the first vertex 4 cell, we do not record these in the vertex lineage (more precisely $q_i\neq q_{i+1}$).}

Now for any root to target path $p\in \mathcal{P}_{1,N}$, that is a path of the form $p=(1,p_2,p_3,\ldots ,N)$, the first time that $p$ is traversed can be defined as 
\begin{equation}\label{def_pathtime}
T(p)=\inf\{t\geq 0: X(p,t)>0\}.
\end{equation}
Observe that path $p$ populating the target first is equivalent to $T(p)=T$.
The question of which path initiated the target only makes sense if the target is eventually populated. To ensure that this occurs we condition on a surviving vertex 1 population, and hence let 
$$
\pr_1(\cdot)=\pr(\cdot\,|\mbox{vertex 1 population survives}).
$$
\mic{Recall that the vertex 1 population surviving is used as a proxy for $T$ being finite (the target is eventually populated), as discussed in the paragraph following \eqref{eqn_timeapprox}.} We can now state the answer to our second question, which is a special case of Theorem \ref{thm_paths} presented below. The probability that the target is populated via path $p$ is simply the path weight of $p$, suitably normalised, that is
\begin{equation}\label{eqn_pathsthm}
\pr_1(T(p)=T)\approx\frac{w(p)}{\sum_{q\in \mathcal{P}_{1,N}}w(q)}=\frac{w(p)}{\phi_{N}}
\end{equation}
for small target seeding transition rates. One may ask not only if a particular path populated the target, but whether the target was initiated from a given set of paths, for an example see Section \ref{sec_twodrugs}. In this case the probability that the target is populated via a particular set of paths is given by summing the normalised path weights over each path in the set.

\begin{figure}[t]
	\centering
	\begin{subfigure}[t]{.45\textwidth}
		\centering
		\resizebox{7cm}{4cm}{
			\begin {tikzpicture}[-latex ,auto ,node distance =3.1 cm and 2cm ,on grid ,
			semithick ,
			state/.style ={ circle ,top color =white , bottom color = processblue!20 ,
				draw,processblue , text=blue ,minimum size =1 cm}]
			\node[state,label={[label distance=0.12cm]90:\normalsize root}] (1)  at (0,0){1};
			\node[state] (2) at (2,2) {2};
			\node[state,label={[label distance=0.12cm]90:\normalsize target}] (3) at (4,0) {3};
					
			\path (1) edge[thick,dashed,red] node[above =0.15 cm,rotate=45,font={\small}] {$\nu$} (2);
			\path (2) edge[thick,dashed,red] node[right =0.3 cm,above=.08 cm,rotate=-45,font={\small}] {$\nu$} (3);
			\path (1) edge node[right =0.15 cm,above=.01 cm,font={\small}] {$\nu^2$} (3);
			\draw[-,line width=.4mm,dashed,red] (3.2,3) to (3.8,3);
			\draw[-] (3.2,2.5) to (3.725,2.5);
		\node[] at (4.3,3) {$p^{(1)}$};
		\node[] at (4.3,2.5) {$p^{(2)}$};

		\end{tikzpicture}
	}
	\caption{}
	\label{fig_valleyfaster}
\end{subfigure}\quad
\begin{subfigure}[t]{.45\textwidth}
	\centering
	\includegraphics[width=7cm,height=5cm]{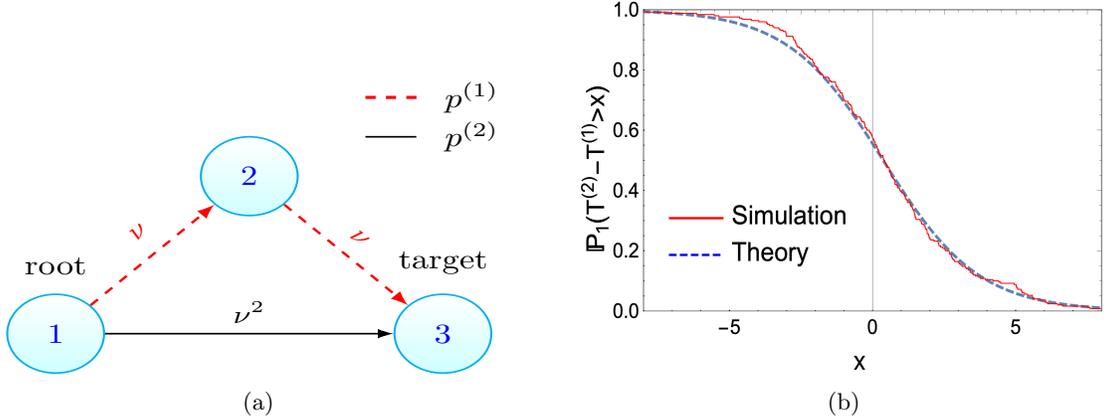}
	\caption{}
	\label{fig_sim}
\end{subfigure}
\caption{(a) Crossing valleys can be faster. Despite the red-dashed path containing a fitness valley \mic{(cells at vertex 2 have a reduced growth rate compared to the root population)}, the target population is more likely to arise via this path if $\growthone-\lambda(2)<1$, \mic{as derived in \eqref{eqn_valleyfaster}. This is confirmed by the simulations displayed in (b), where $\growthone - \lambda(2)=0.8<1$ and the Monte-Carlo estimate for $\pr_1(T^{(2)}>T^{(1)})$ (the probability path 1 seeds the target first) is 0.58, while the analytic result yields 0.56}.  (b) Stochastic simulations of system in (a) vs theory \eqref{eqn_valleyfaster}. Parameters $\birthone = 0.9,\,\deathone=0.3,\,\growthone=0.6,\,\alpha(2)=0.2,\,\beta(2)=0.4,\,\lambda(2)=-0.2,\,\nu=0.1$, runs = 250.}
\label{fig_valleys}
\end{figure}

If the vertex 1 population avoids extinction then $T(p)$ will be finite for each root to target path $p$. Therefore instead of only asking which path populated the target first, one might be interested in whether a particular path is traversed significantly faster. For example, does multidrug resistance obtained via one of the paths in Fig.\ \ref{fig_example} occur $t$ days before resistance from the other path. In order to also consider this case, for a root to target path $p$, we define
\begin{equation}\label{def_notptime}
T(\neg p)=\min\{T(q):q\in \mathcal{P}_{1,N}\backslash\{p\}\},
\end{equation}
that is the first time the target is reached along any path except $p$.

We can now quantify the probability of reaching the target via path $p$ more than $t$ time units before any other path. To avoid discussing a technical assumption, which excludes pathological scenarios, the theorem is stated in approximate form. The more precise version can be found in the supplementary material (\mic{Theorem \ref{thm_paths2} in Section \ref{sec_SI_pathproof}}). 
\begin{theorem}\label{thm_paths} 
	For small target seeding transition rates, and $t\in \mathbb{R}$,
	\begin{equation}
	\pr_1(T(\neg p)-T(p)>t)\approx \frac{w(p)}{w(p)+e^{\growthone t}\sum_{q\in \mathcal{P}_{1,N}\backslash\{p\}}w(q)}.
	\end{equation}
\end{theorem} 
Note that letting $t=0$ gives \eqref{eqn_pathsthm}. In the running example of a path graph we cannot ask which path seeded the target vertex. However we briefly illustrate the usefulness of the above theorem by considering the question of whether the target vertex is populated via a fitness valley \mic{(where by fitness valley we mean a path such that cells at the internal vertices have smaller growth rates than the root population)} .

Let us consider the case depicted in Fig.\ \ref{fig_valleyfaster} where two paths exist to the target, a direct path and an indirect path. Here $N=3$, with paths $p^{(1)}=(1,2,3),\,p^{(2)}=(1,3)$ and $\nu=\nu(1,2)=\nu(2,3),\,\nu(1,3)=\nu^2$. Naively it may be expected that as the product of transition rates along both these paths are equal, and the indirect path ($p^{(1)}$) contains a vertex with a fitness cost, the direct path ($p^{(2)}$) is more likely. However, for small $\nu$, Theorem \ref{thm_paths} informs us that the target is populated via the indirect path $t$ time units before the direct path, with probability
\begin{equation}\label{eqn_valleyfaster}
\pr_1(T^{(2)}-T^{(1)}>t) \approx [1+e^{\growthone t}(\lambda-\lambda(2))]^{-1}
\end{equation}
where $T^{(i)}=T(p^{(i)})$. Thus the indirect path is more probable, that is $\pr_1(T^{(1)}<T^{(2)})>\pr_1(T^{(2)}<T^{(1)})$, if $\growthone-1<\lambda(2)$.

Further applications of the results presented thus far will be given in Section \ref{sec_applications}. Before this, we discuss the case when $G$ is cyclic, and some extensions to the initial model for which our results still hold.

\subsection{Extensions }\label{sec_cyc}

\begin{figure}[b!]
	\begin{subfigure}[t]{.45\textwidth}
		\centering
		
		\resizebox{5cm}{3cm}{
			\begin {tikzpicture}[-latex ,auto ,node distance =2.1 cm and 1cm ,on grid ,
			semithick ,
			state/.style ={ circle ,top color =white , bottom color = processblue!20 ,
				draw,processblue , text=blue ,minimum size =.5 cm}]
			\node[state,label={[label distance=0.12cm]90:\normalsize root}] (1)  at (0,0){1};
			\node[state] (2) [right of =1] {2};
			\node[state] (3) [above of = 2] {3};
			\node[state,label={[label distance=0.12cm]90:\normalsize target}] (4) [right of = 2] {4};
			
			\path (2) edge[thick,dashed,red] node[right =0.4 cm,above=.08 cm,rotate=-90,font={\normalsize}] {$\nu_2$} (3);
			\path (1) edge node[right =0.15 cm,above=.01 cm,font={\normalsize}] {$\nu_1$} (2);
			\path (2) edge node[right =0.15 cm,above=.01 cm,font={\normalsize}] {$\nu_1$} (4);
			\path (3) edge[thick,dashed,red] node[left =0.4 cm,above=.08 cm,rotate=90,font={\normalsize}] {$\nu_2$} (2);
			
		\end{tikzpicture}
	}
	\caption{ }
	\label{fig_walkfaster}
\end{subfigure}
\begin{subfigure}[t]{.45\textwidth}
	\centering
	\resizebox{7cm}{4cm}{
		
		\begin{tikzpicture}[-latex ,auto ,node distance =2.8 cm and 2cm ,on grid ,
		semithick ,
		state/.style ={ circle ,top color =white , bottom color = processblue!20 ,
			draw,processblue , text=blue , minimum width =.7 cm}]	
		
		\node[state,label={[label distance=0.12cm]97:\large root}] (v1) at (-2.5,0) {1};
		
		\node[state,label={[label distance=0.12cm]90:\large $\lambda_1(2)$}](v2) [above right of = v1] {2};

		\node[state,label={[label distance=0.12cm]270:\large $\lambda_2(2)$}](v24) [below right of = v1] {3};
		\node[state]  (v3) [above right of = v24] {4};
		\node[state,label={[label distance=0.12cm]85:\large target}] (v4) [ right of = v3]  {5};

		\path (v1) edge node[above =0.0 cm,rotate=45] {$\nu(1,2)\pi_{2}(1)$} (v2);
		\path (v1) edge node[below =0.0 cm,rotate=-45] {$\nu(1,3)\pi_{2}(2)$} (v24);
		
		\path  (v2) edge node[above=0.0 cm,rotate=-45] {$\nu(2,3)$} (v3);
		\path  (v24) edge node[below=0 cm, rotate=45] {$\nu(2,3)$} (v3);
		\path (v3) edge node{$\nu(3,4)$} (v4);
		\end{tikzpicture}
	}
	\caption{}
	\label{fig_distfit}
\end{subfigure}
\caption{\mic{ Example systems for which our extended results, Section \ref{sec_cyc}, are useful.} (a) Walks containing cycles can seed the target. The walk of length 4 that traverses the red edges with transition rate $\nu_2$ can still be faster than the length 2 path. See \protect\eqref{eqn_walkfaster}. (b) Incorporating a distribution of fitnesses. Take the path graph from Fig.\ \ref{fig_lineargraph1}. Suppose that instead of a fixed $\lambda(2)$, after the first transition cells have fitness $\lambda_1(2)$ or $\lambda_2(2)$ with probabilities $\pi_{2}(1),\,\pi_{2}(2)$, with $\lambda_i(2)<\growthone$. 	Then instead of the original path graph, we consider the enlarged graph shown here, where vertex 2 is replaced by two vertices with associated growth rates $\lambda_1(2),\,\lambda_2(2)$ . The transition rates to these new vertices are $\nu(1,2)\pi_{2}(i)$, where $\nu(1,2)$ is the transition rate between vertices 1 and 2 in the original setting.  All results can then be applied to the new enlarged graph. }
\end{figure}
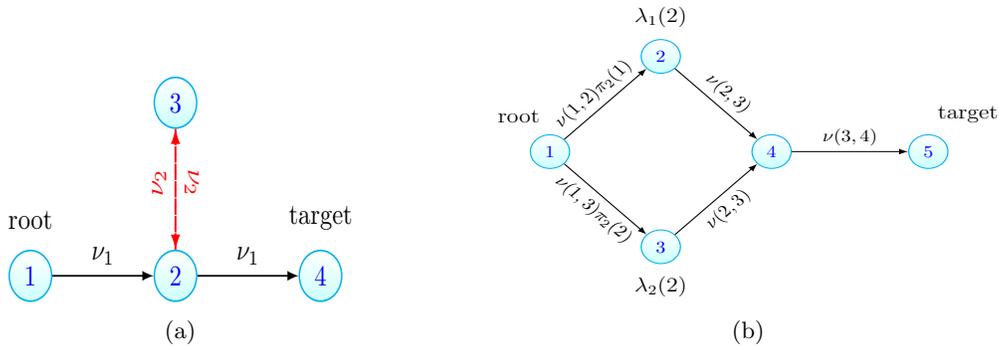

\subsubsection{Cyclic graphs}
In the above exposition we have assumed that $G$ is acyclic, that is no path exists which contains cycles. We now briefly summarise how the results are altered if $G$ is cyclic, with details given in the supplementary material. We will however maintain that the target is a sink vertex (has no outgoing edges) and the root is a source vertex (has no incoming edges). \mic{The condition on the root is perhaps unnatural, as we now allow back transitions between vertices except to the root, but is required for technical reasons which are discussed in the supplementary material (see the remarks following the proofs of Theorem \ref{thm_Zk2} and Proposition \ref{prop_acyctocyc}). We aim to drop this condition in future work.}.

First we introduce the $N-1\times N-1$ matrix $A$ with entries 
\begin{equation}a_{ij}=
\begin{cases}
\lambda(i)\quad j=i
\\
\nu(j,i)\quad \mbox{if } (j,i)\in E
\\
0\quad \mbox{otherwise}
\end{cases}
\end{equation}
and we denote the largest real eigenvalue of $A$ as $\lambda^*$.

In order that a modified version of Theorem \ref{thm_timesacyclic} holds it is required that $\lambda^*=\lambda$ and $\lambda^*$ is a simple eigenvalue. This condition is guaranteed for acyclic graphs, and will also be satisfied if the transition rates throughout the graph are small enough (see Lemma \ref{lem_sufcond} in the supplementary material). Assuming this eigenvalue condition, Theorem \ref{thm_timesacyclic} holds but with $\phi_N$ replaced by summing particular elements of the $\growthone$-associated eigenvector of $A$ (see Theorem \ref{thm_timesgeneral}). Theorem \ref{thm_Zk1}, which concerns the population growth at the initial and intermediate vertices and is presented below in Section \ref{sec_methods}, is similarly modified under the same condition (see Theorem \ref{thm_Zk2}).

Turning to Theorem \ref{thm_paths}, we can now consider walks that contain cycles as opposed to only paths in the graph $G$. If we consider a particular set of root to target walks of a finite length, analogously to Theorem \ref{thm_paths}, we can ask whether the target is reached via one particular walk more than $t$ units before any other walk in the specified set. A near identical result to Theorem \ref{thm_paths} is obtained (see Proposition \ref{prop_acyctocyc} in the supplementary material). However the path weights should now be replaced by the walk weights (defined analogously to the path weight; the product of the transition rates along the walk divided by the fitness costs along the walk). In particular for two given walks of finite length we can still approximate the probability the target population arises via one of the walks versus the other.
%

For example we can use this to explore when a walk, containing a cycle, will seed the target before a path. Consider the graph and edge parameters displayed in Fig.\ \ref{fig_walkfaster}. Let $p=(1,2,4)$ and $\walk=(1,2,3,2,4)$.  Despite the extra length, by using the modified Theorem \ref{thm_paths}, for small $\nu_1$ we see
\begin{equation}\label{eqn_walkfaster}
\pr_1(T(\walk)<T(p))\approx \left(1+\frac{(\growthone-\lambda(2)) (\growthone-\lambda(3))}{\nu_2^2}\right)^{-1}.
\end{equation}
Hence the target can be seeded via the walk, if $\nu_2$ is large enough. On the other hand, by using the modified Theorem \ref{thm_paths} it can be seen that when all the transition rates are small in the system, the target population will arise from a path and not a walk containing cycles (see Proposition \ref{prop_nowalks}) This further justifies our focus on the acyclic case. 


\subsubsection{Distribution of fitness effects}
Let us now expand on the comments in Section \ref{sec_model} where we claimed that having a finite distribution of fitness effects is covered by our model. \mic{This may be of interest, for example, when a cell acquires drug resistance via a mutation, but the cost of resistance is stochastic}. More generally, suppose we wish to consider cells transitioning from vertex $x$ to vertex $y$ but require that the new growth rate of cells to be $\lambda_i(y)$ with probability $\pi_{y}(i)$, where $\pi_{y}(i)$ is a finite distribution and every $\lambda_i(y)<\lambda$. Then we can simply replace the vertex $y$ with new vertices associated to each of the $\lambda_i(y)$ and have cells transition from vertex $x$ to the new vertices at rate $\nu(x,y)\pi_y(i)$. We may freely apply the results of Section \ref{sec_results} to the graph containing the new vertices and edges. See figure \ref{fig_distfit} for an example.

%

 \subsubsection{Alternative offspring distribution}\label{sec_genoff}
We have focused on the setting where cells can divide, die or transition, however Theorem \ref{thm_paths} would also hold for a slightly more general branching process. Keeping the transition process as it is, we can have a vertex $x$ cell \mic{being replaced by} $j$ vertex $x$ cells at a rate $\alpha(x)\sigma_{x,j}$ where for each $x$ between $1$ and $N-1$, $(\sigma_{x,j})_{j\geq 0}$ is a probability distribution with finite second moment and mean $\bar{\sigma}_x$. We now have no need for $\beta(x)$ as cell death is contained in $\sigma_{x,0}$. This scenario may be relevant to the case of viral dynamics, where an infected cell can release multiple virions upon bursting \cite{Pearson:2011}. In this setting if we redefine $\lambda(x) =\alpha(x)(\bar{\sigma}_x-1)$ then the path distribution will be unchanged. This is due to the fact that the key result of \cite{Janson:2004} used in the proof of Theorem \ref{thm_Zk1} holds then also. However Theorem \ref{thm_timesacyclic} is no longer true in this setting. This is due to a lack of understanding of the large time behaviour of the root population (see the remarks following the proof of Theorem \ref{thm_Zk2} in the supplementary material).

\section{Discussion}\label{sec_discussion}
Before discussing related work and summarising this study, we demonstrate how to apply the results of Section \ref{sec_results} on some applications. The approximate formulas given in Section \ref{sec_results} will be our key tools, and we now briefly remark on these.

The results of Section \ref{sec_results} hold when the final transition rates in each path tends to 0. 
Therefore in our applications all transition rates, in particular the mutation and migration rates discussed below, will be taken to be small and statements are to be interpreted as approximations that hold true in this limiting regime. 
%
\mic{We also note that in our model, when transitions represent migrations, they also occur at cell divisions, $(x)\rightarrow (x),(y)$, instead of at any time, as $(x)\rightarrow (y)$. This might appear strange from a biological point of view, but the former formulation for transitions has been chosen, as it simplifies the mathematical treatment. More importantly, we expect that for small transition rates these formulations lead to very similar target hitting times. Indeed, simulations support this claim, as presented in Fig.\ \ref{fig_alt_trans}. In all applications considered we shall neglect the role of back transitions. This is for simplicity and is in keeping with the previous works that we compare with. The effect of (finitely many) back transitions could be included by using our results on cyclic graphs.}

\subsection{Applications}\label{sec_applications}
First we consider the impact of imperfect drug penetration on the emergence of resistance. Two recent publications have explored how resistance spreads in this setting, and have shown that poor penetration can accelerate resistance \cite{Fu:2015,Moreno-Gamez:2015}. We are able to recover and extend \mic{some of} their findings in a rigorous fashion. Next, the ordering by which resistance-causing mutations accrue is investigated, firstly in the setting of cancer and then in bacterial infections. In the case of bacterial infections we examine how the risk of multidrug resistance depends on mutation rates, recovering simulation results in \cite{Ford:2013}. 
\mic{A particular aim of this section is to illustrate how to apply the results of Section \ref{sec_results}, and so, for the readers' convenience, we have collected the key formulas in Table \ref{table_keyformulas}.}

\renewcommand{\arraystretch}{2}

\begin{table}[h!]
	\centering
		\caption{\mic{Key formulas needed for Section \ref{sec_applications}. These approximations will be valid for small transition rates. The approximation for $t_{1/2}$ assumes a large initial number of cells at the root, $z\gg 1$. For small $z$ see Corollary \ref{cor_medpath} in the supplementary material.}}.
	\begin{tabular}{|| l | l | l ||} 
		\hline
		Description & Formula & Reference \\ [0.5ex] 
		\hline
		Weight of path $p$ & 	$w(p)=\nu(p_1,p_2)\prod_{i=2}^{l}\frac{\nu(p_i,p_{i+1})}{\growthone-\lambda(p_i)}$ & Eq. \eqref{def_pathweight}
		\\ 
		Total weight of target & 
		$\phi_N = \sum_{p\in \mathcal{P}_{1,N}}w(p)$ & Eq. \eqref{def_totalweight} \\
		Distribution of target hitting time & 	$
		\pr(T>t)\approx\left(\frac{\growthone/\birthone}{1+e^{\lambda t}\phi_N \alpha/\lambda^2}+\deathone/\birthone\right)^z$ & Eq. \eqref{eqn_timeapprox} \\
		Median of target hitting time & 
		$t_{1/2}\approx \frac{1}{\growthone}\log \frac{\growthone}{z\phi_{N}}$ & Eq. \eqref{eqn_medtime} \\
		Probability target populated via path $p$ & $
		\pr_1(T(p)=T)\approx \frac{w(p)}{\phi_{N}}$ & Eq. \eqref{eqn_pathsthm} \\ [1ex] 
		\hline
	\end{tabular}
	\label{table_keyformulas}
\end{table}

\subsubsection{Imperfect drug penetration: monotherapy}\label{sec_singledrug}
Our first two applications concern imperfect drug penetration and so the language introduced here will hold in the next section also. \mic{The first scenario we examine is exactly the same as that considered in \cite{Fu:2015}.}

We consider a pathogenic population,  with cancer cells or bacteria as example populations, being treated with one or more drugs with imperfect penetration profiles. The imperfect penetration results in low drug concentration in spatial locations, such that cells in these locations may still proliferate. We term the cumulation of these spatial regions the \textit{sanctuary}. Sanctuaries have been observed in both bacterial infections \cite{Deresinski:2009,Warner:2006} and tumours \cite{Primeau:2005}. We consider resistance to have arisen when cells with sufficiently many mutations such that all therapies are ineffective come to exist in areas where the drugs have penetrated. The possibility of resistance occurring prior to treatment is excluded and we consider only the case when resistance arises from the sanctuary. Further, we do not consider cells acquiring mutations via gene transfer.

Throughout this section and the next we shall be comparing paths composed of mutation and migration events. Growing cells in the sanctuary will be the vertex 1 cells from Section \ref{sec_model} and thus have birth and death rates $\alpha$, $\beta$, with growth rate $\growthone$. The lack of exposure to drug motivates $\growthone>0$. We shall assume all cells have birth rate $\alpha$. Across all cell types, i.e. regardless of location or drug presence, the per cell mutation and migration rates will be denoted $\nu$ and $m$ respectively.  We fix the initial number of cells in the sanctuary to be $z$.  \mic{We will often suppose that $z$ is dependent on how penetrative a particular drug is (as this will affect the sanctuary size)}.



\mic{ In the case of monotherapy,} there are two paths from which resistance might arise. Firstly, upon replication a cell in the sanctuary may produce a mutant able to grow in the presence of the drug. Then a mutated cell in the sanctuary can migrate to the drug compartment. We term this the mutation-migration path. Analogously we have the migration-mutation path which involves a susceptible cell migrating from the sanctuary, from whose progeny a resistant mutant emerges.  We suppose a resistance-conferring mutation carries a fitness cost which will be denoted $s$. Resistance costs may not always exist but are well documented and often assumed in viruses, bacteria and cancer \cite{Andersson:2010,Hughes:2015,Enriquez-Navas:2016}.  Due to the fitness cost of resistance the death rate of mutated cells in the sanctuary is increased to $\beta+s$,  while the death rate of migrated susceptible cells in the drug compartment is $\beta+d$, with $s,d>0$. 

The first quantity to determine is the probability of escaping resistance, and its dependence on the initial number of cells in the sanctuary. The approximation of \eqref{eqn_proboccurs} shows that the probability resistance never occurs decreases exponentially with the initial number of cells in the sanctuary as $\pr(\text{resistance never occurs})\approx (\beta/\alpha)^{z}$. For the rest of the discussion we will assume that resistance does occur. Further we suppose $z\gg 1$, so that we may use the approximate form for $h(z)$ given in \eqref{eqn_hasym}. The case when $z$ is small could be treated similarly using the exact expression for $h(z)$.

We now ask when, and via which path, will resistant cells come to exist in the drug compartment? From Section \ref{sec_results} we know the key quantities required to answer these questions are the path weights. For this system these are given by
\begin{equation}\label{eqn_mono_pathweights}
w(\mbox{mutation-migration path})=\frac{\nu m}{s},\quad w(\mbox{migration-mutation path})=\frac{\nu m}{d}.
\end{equation}
If we now let $t_{1/2}$ be the median time resistance occurs, then 
using the path weights with \eqref{eqn_medtime}, we immediately have
\begin{equation}\label{eqn_sdmed}
t_{1/2} \approx \frac{1}{\growthone}\log\frac{\growthone^{2}}{\alpha}-\frac{1}{\growthone}\log \left(\frac{\nu m}{s}+\frac{\nu m}{d}\right)+\frac{1}{\growthone}\log\frac{\alpha}{z \growthone}
\end{equation}
Of note is that the sanctuary size, mutation and migration rate all have equivalent impact. This simple result can now be explored under different biological scenarios.

For example, suppose we are currently using drug $A$ with fitness cost $d^{(A)}$ and penetration profile leaving $z^{(A)}$ initial cells in the sanctuary. We consider instead using a second drug which results in sanctuary size $z^{(B)}$ and has fitness cost $d^{(B)}$. Let the median resistance times under the first and second drug be denoted $t^{(A)}_{1/2}$ and $t^{(B)}_{1/2}$ respectively. Then, using \eqref{eqn_sdmed} and examining $t^{(B)}_{1/2}>t^{(A)}_{1/2}$, we find we should switch to drug $B$, if

\begin{equation}\label{ineq_singledrug}
\frac{z^{(A)}}{z^{(B)}}>\frac{1+ s/ d^{(B)}}{1+s/d^{(A)}}.
\end{equation}
This inequality is illustrated in Fig.\ \ref{fig_sd}.

\begin{figure}
	\centering
	\begin{subfigure}[t]{.45\textwidth}
		\centering
		\resizebox{8cm}{6cm}{
		\begin {tikzpicture}[-latex ,auto ,node distance =3.5 cm and 2cm ,on grid ,
	semithick ,
	state/.style ={ circle ,top color =white , bottom color = processblue!20 ,
		draw,processblue , text=blue ,minimum size =1.3 cm,font={\Large}}]
			\node[state] (1)  at (0,0){0};
			\node[state] (2) [right of =  1] {s};
			\node[state] (3) [below of =  1] {d};
			\node[state] (4) [right of =  3] {s};
	
			\node[rotate=45,above,font={\Large}] (16)  at (.5,1.5){\makebox[.4cm][b] Sensitive};
			\node[rotate=45,above,font={\Large}]  (17)  at (3.6,1.5){\makebox[.4cm][b] Resistant};

			\node[align=left]  at (-2,0){\Large Sanctuary};
			\node[align=left] at (-2,-3.1){\Large Drug };

			
			
			%
			\path (1) edge node[above =0.15 cm] {Mutation-$\nu$} (2);
			\path (2) edge node[right =0.3 cm,above=.08 cm,rotate=-90] {Migration-$m$} (4);
			\path (1) edge node[right =0.3 cm,above=.08 cm,rotate=-90] {Migration-$m$} (3);
			\path (3) edge node[above =0.15 cm] {Mutation-$\nu$} (4);
		\end{tikzpicture}
	}
	\caption{}
	\label{fig_singledruggraph}
	\end{subfigure}\quad
\begin{subfigure}[t]{.45\textwidth}
	\centering
	\includegraphics[width=6cm,height=6cm]{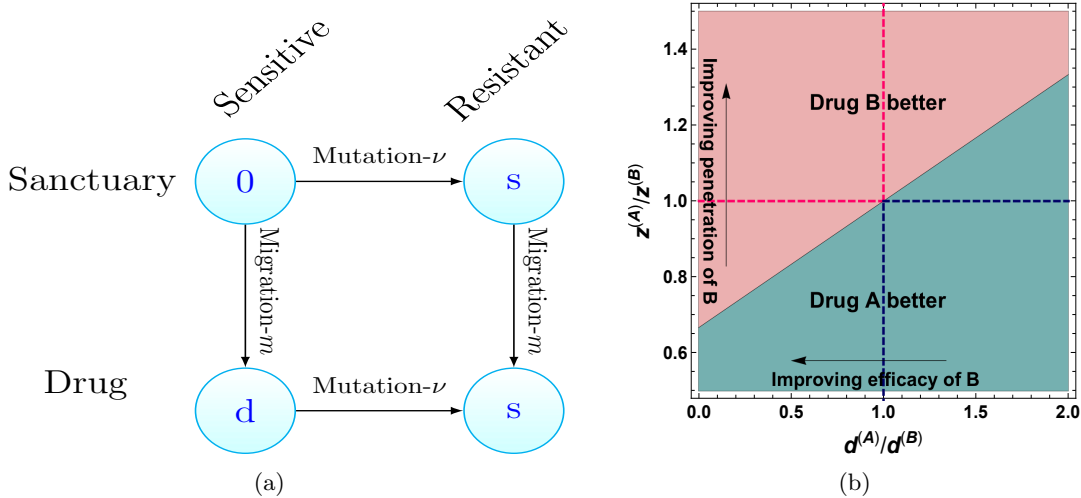}
	\caption{}
\label{fig_sd}
\end{subfigure}
\caption{(a) In the case of monotherapy, there are two paths by which drug resistance cells can arise in the region containing the drug. Horizontal edges represent changes in genotype (gaining resistance) whereas vertical edges represent changes in spatial location. The values of vertices are the fitness cost cells experience in that state and $\mu,\, m$ are the per cell mutation and migration rates. (b) Comparison of two drugs with differing efficacies and penetration profiles using \eqref{ineq_singledrug}. Arrows pertain to characteristics of drug $B$ relative to drug $A$. Top left and bottom right quadrants illustrate region where drugs are trivially better, i.e. have superior penetration and efficacy. Parameters: $d^{(A)}=10^{-2},\,s=5\times 10^{-3}$.}
\end{figure}

Note as $d^{(B)}\rightarrow \infty$, implying drug $B$ is far more effective,  the condition tends to $z^{(A)}/z^{(B)}>(1+s/d^{(A)})^{-1}$. This is due to the fact that when $d^{(B)}$ is large resistance will always arise via the mutation-migration path (which we show below), so this provides an upper bound to the gain that can be obtained from increasing the drug efficacy. The above is under the assumption resistance will occur but note the probability of no resistance would decrease exponentially with larger sanctuary size.

Having discussed the time resistance will occur we now ask by which of the two paths it will occur. Let $T^{\nu m}$ be the first time resistance arises via the mutation-migration path and $T^{m \nu}$ be defined analogously for the migration-mutation path. Then Theorem \ref{thm_paths} gives us
\begin{equation}
\pr(T^{\nu m}< T^{m \nu}) \approx \frac{d}{s+d}.
\end{equation}
Hence resistance arising via the mutation-migration path is more probable, that is the above probability is greater than $1/2$, if $d>s$. As most drugs are designed to at least halt cell growth, while many resistant cell lines can still proliferate, we expect this condition to be satisfied for most drugs. Note that neither $m$ nor $\nu$ appear in the above expression as the numerator of both path weights is equal, see \eqref{eqn_mono_pathweights}, and hence cancels.

A similar question was asked in \cite{Fu:2015}. There they considered an approximate stochastic process and sought the parameter regime such that $T^{m \nu}$ stochastically dominated $T^{\nu m}$ under the approximate process. For the full process considered here, we can ask the same question but on the limit of the centred (in the sense of Theorem \ref{thm_timesacyclic}) resistance times. The resulting condition for stochastic dominance is again $d> s$. The difference in the condition given here ($d>s$) and that stated in equation (28) in \cite{Fu:2015} exists as there $\nu$ is defined as the probability of mutation per birth event. Accounting for this leads to the same condition.

%
%
%

\subsubsection{Imperfect drug penetration: combination therapy}\label{sec_twodrugs}

There are at least two reasons for extending our analysis to include two drugs, say drugs $A$ and $B$, being used in tandem. Firstly, combination therapies are widely implemented in HIV and bacterial infections, and there is a growing appreciation for their use in cancer \cite{Cihlar:2016,Tamma:2012,Bozic:2013}. Secondly, interesting new phenomena emerges, especially when a region still exists with only one drug present. \mic{The work presented in this section is inspired by the model and questions of \cite{Moreno-Gamez:2015}.}

We will only consider the case where the penetration profile for drug $B$ is a subset of drug $A$, and thus a single drug region exists only for drug $A$. This allows us to clearly demonstrate the impact of unequal penetration profiles. In this setting there are 12 possible states that a cell may inhabit, defined by which of the drugs the cell is exposed to and the resistance profile it has obtained. These states are illustrated in Fig.\ \ref{fig_twodrugstates}. We further assume that the mutation rate bringing resistance to either drug is $\nu$. This is for simplicity only, and dealing with differing mutation rates is straightforward. \mic{Mutations that confer resistance will again have a fitness cost $s$, and sensitive cells coming into contact with drugs invokes a fitness cost $d$. These fitness costs increase the cellular death rate additively. The increase to the death rate for a cell in any given state is displayed as the vertex labels in Fig.\ref{fig_twodrugstates}. Drug interactions, mutation specific fitness costs, and epistasis are neglected but may easily be included.}

 To investigate the effect of the single drug region, we introduce further notation. Let $n_{\rm{Tot}}$ be the number of cells that may reside in our system of interest (e.g.\ the entire patient or tumour) in the absence of drug. Further, let the regions where both drugs act, the single drug (drug $A$) acts and the sanctuary have their respective cell capacities denoted by $n_{\rm{DD}},\,n_{\rm{D}},$ and $n_{\rm{S}}$. We assume that all these capacities are large, and hence we may still investigate the system with our branching process model. These capacities enter the dynamics of the process in the following manner. When a cell leaves the region it is in, at rate $m$, it now may migrate to one of two regions. Of these two regions, we specify that it migrates to a particular region at a frequency proportional to the regions capacity. For example a cell in the sanctuary migrates to the double drug region at rate $m n_{\rm{DD}}/(n_{\rm{D}}+n_{\rm{DD}})$. Resistance mutations and the presence of the drug has the same additive effect as in Section \ref{sec_singledrug}.

\begin{figure}
\centering
\begin{subfigure}[t]{.45\textwidth}
	\centering
	\resizebox{8cm}{6cm}{
	\begin {tikzpicture}[-latex ,auto ,node distance =3.1 cm and 2cm ,on grid ,
semithick ,
state/.style ={ circle ,top color =white , bottom color = processblue!20 ,
	draw,processblue , text=blue ,minimum size =1.3 cm,font={\Large}}]
\node[state] (1)  at (0,0){0};
\node[state] (2) [right of =  1] {s};
\node[state] (3) [right of =  2] {s};
\node[state] (4) [right of =  3] {2s};
\node[state] (5) [below of =  1]{d};
\node[state] (6) [right of = 5] {s};
\node[state] (7) [right of =  6] {s+d};
\node[state] (8) [right of =  7] {2s};

\node[state] (9) [below of =  5]{2d};
\node[state] (10) [right of = 9] {s+d};
\node[state] (11) [right of =  10] {s+d};
\node[state] (12) [right of =  11] {2s};

\node (13) [above left of = 1] {};
\node (14) [above right of = 4] {};
\node (15) [below left of = 9] {};

\node[rotate=45,above,font={\LARGE}] (16)  at (.5,2){\makebox[.4cm][b] Sensitive};
\node[rotate=45,above,font={\LARGE}]  (17)  at (3.6,2){\makebox[1.5cm][b] A-Resistant};
\node[rotate=45,above,font={\LARGE}] (17)  at (3.6+3.1,2){\makebox[1.5cm][b]B-Resistant};
\node[rotate=45,above,font={\LARGE}] (17)  at (3.6+2*3.1,2){\makebox[1.5cm][b] AB-Resistant};

\node[align=left]  at (-2.4,0){\LARGE Sanctuary};
\node[align=left] at (-2.4,-3.1){\LARGE Drug A};
\node[align=left] at (-2.4,-6.2){\LARGE Drug A \& B};


%
\path (1) edge[red,thick,dashed] node[above =0.15 cm] {} (2);
\path (3) edge node[above =0.15 cm] {} (4);
\draw [->] (1) to[bend left] node[auto] {} (3);
\draw [->] (2) to[bend right] node[auto] {} (4);

\path (1) edge node[above =0.15 cm] {} (5);
\path (2) edge[red,thick,dashed] node[above =0.15 cm] {} (6);
\path (3) edge node[above =0.15 cm] {} (7);
\path (4) edge node[above =0.15 cm] {} (8);

\draw [->] (1) to[bend right] node[auto] {} (9);
\draw [->] (2) to[bend right] node[auto] {} (10);
\draw [->] (3) to[bend right] node[auto] {} (11);
\draw [->] (4) to[bend right] node[auto] {} (12);

\path (5) edge node[above =0.15 cm] {} (6);
\path (7) edge node[above =0.15 cm] {} (8);
\draw [->] (5) to[bend left] node[auto] {} (7);
\draw [->,red,thick,dashed] (6) to[bend right] node[auto] {} (8);

\path (5) edge node[above =0.15 cm] {} (9);
\path (6) edge node[above =0.15 cm] {} (10);
\path (7) edge node[above =0.15 cm] {} (11);
\path (8) edge[red,thick,dashed] node[above =0.15 cm] {} (12);

\path (9) edge node[above =0.15 cm] {} (10);
\path (11) edge node[above =0.15 cm] {} (12);
\draw [->] (9) to[bend left] node[auto] {} (11);
\draw [->] (10) to[bend right] node[auto] {} (12);

\end{tikzpicture}
}
	\caption{}
	\label{fig_twodrugstates}
\end{subfigure}\hspace{1cm}
\begin{subfigure}[t]{.45\textwidth}
	\centering
	\includegraphics[width=8cm,height=6cm]{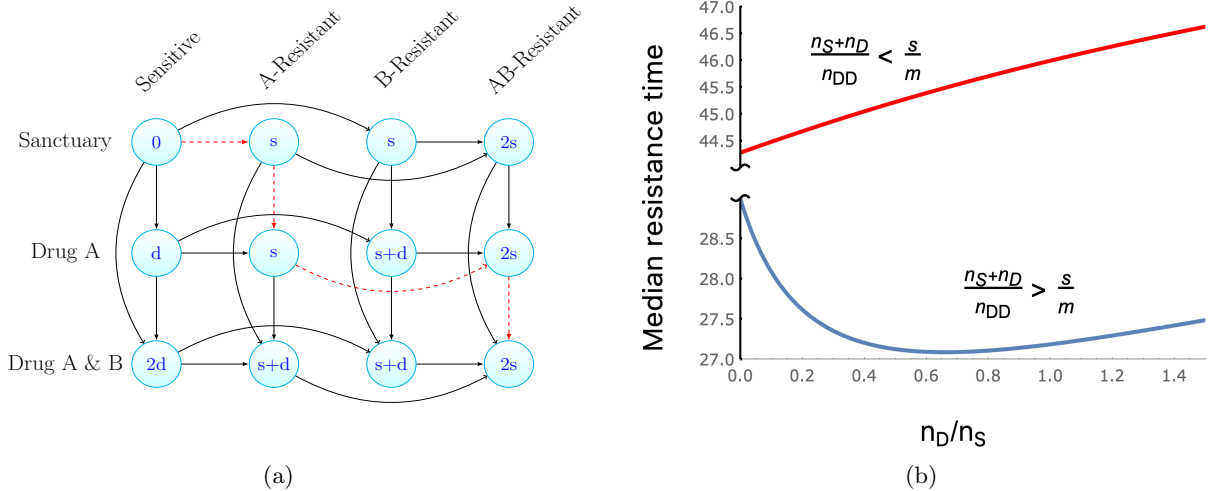}
	\caption{ }
	\label{fig_twodrugstime}
\end{subfigure}
\caption{(a) Illustration of the paths to multidrug resistant cells arising in the region containing both drugs $A$ and $B$, as in Fig.\ \ref{fig_singledruggraph}. Note a spatial location exists containing only drug $A$. \mic{The red, dashed path is an example of stepwise evolution, where multidrug resistant cells come to exist in the sanctuary via the region containing only a single drug. The vertex labels display the fitness cost (increase to the cellular death rate) of cells in that state.}	(b) \mic{Exploration of the median time to resistance as the size of the drug $A$ region increases relative to the sanctuary size. Here we plot $t_{1/2}$, as given in Table \ref{table_keyformulas}, for the system displayed in (a). We observe that, for a subset of the parameter space, the existence of a single drug region can decrease the time to resistance, despite diminishing the sanctuary size. If both drugs penetrate the majority of the system and the drugs have a greater effect than the cost of resistance, this subset of parameter space is identified by the inequality \eqref{cond_decreasetime}. Within this parameter subset, resistance is fastest at a particular size of single drug compartment, given by  \eqref{eqn_mintime}}. Parameters; (top) $\nu=10^{-6},\,m=0.05,\,d=0.9,\,s=10^{-3},\,\growthone=0.4,\birthone=0.3,\,n_{\rm{Tot}}=10^7,\,n_{\rm{DD}}=0.9n_{\rm{Tot}},\,\gamma=10^{-3}$; (bottom) all same except $s=10^{-2}$. }
\end{figure}

 Observe that there are 18 possible paths by which sensitive cells in the sanctuary can eventually produce fully drug resistant cells in the double drug compartment. To see this it is convenient to separate the paths into those that pass through the single drug compartment, which following \cite{Moreno-Gamez:2015} we term stepwise evolution, and those that involve migration from the sanctuary directly to the double drug compartment, which we collectively call direct evolution. For the direct evolution paths there are three transitions needed, mutations yielding resistance to drug $A$ and $B$ and then the migration to the double drug region. From the number of permutations of these events, this yields 6 direct evolution paths. In the stepwise evolution case four transitions are needed, both mutations and then the two requisite migration events. Notice that the migration from the sanctuary to the single drug region must precede the migration from the single drug region to the double drug region. This leads to 12 stepwise evolution paths. An exemplar stepwise evolution path is a mutation to drug $A$,  followed by migration to the single drug region, then a mutation to drug $B$, and finally a migration to the double drug region. The edges of this path are red and dashed in Fig.\ \ref{fig_twodrugstates}. 
 

  \mic{Suppose, momentarily, that the penetration profiles of drugs $A$ and $B$ exactly match (but are still imperfect). In this case, resistance in the double drug region can be obtained only via direct evolution. If we now increase the penetration of drug $A$ only, 12 new possible paths to resistance open up (the stepwise evolution paths of the previous paragraph), accelerating the time until resistance. However, the increased penetration also leads to a reduction in the sanctuary size, which should delay resistance. Hence a trade-off may exist with respect to the penetration profile of drug $A$. }
  To study this we fix $n_{\rm{Tot}}= n_{\rm{S}}+n_{\rm{D}}+n_{\rm{DD}}$ and specify that the initial number of cells in the sanctuary is related to its size via $z=\gamma n_{\rm{S}}$ with $0<\gamma\ll 1$. Further let $t_{1/2}(n_{\rm{D}}/n_{\rm{S}})$ be the median time to resistance with fixed $n_{\rm{D}}/n_{\rm{S}}$. Using the approximation provided by \eqref{eqn_medtime} we can explore the behaviour of $t_{1/2}(n_{\rm{D}}/n_{\rm{S}})$. As can be seen in Fig.\ \ref{fig_twodrugstime} it is possible, \mic{in certain regimes of the parameter space}, that the existence of the single drug compartment accelerates resistance (\mic{by reducing $t_{1/2}$)}. We now seek the conditions under which this is the case. \mic{Brief derivations of the presented formulas below \mref{cond_decreasetime,eqn_mintime,eqn_twodrug_prob} are given in the supplementary material (Section \ref{sec_deriv_twodrug})} ,

 As mentioned previously, it is reasonable to expect the efficacy of either drug is greater than the cost to resistance. Also we may hope that both drugs penetrate the majority of the target system. This motivates taking $s\ll d$ and $(n_{\rm{S}}+n_{\rm{D}})\ll n_{\rm{Tot}}$. Under this limit, by examining the sign of $\frac{d}{d (n_{\rm{D}}/n_{\rm{S}})}t_{1/2}(0)$, we find a simple condition for when increasing the single drug compartment \mic{accelerates resistance}, namely 
\begin{equation}\label{cond_decreasetime}
\frac{n_{\rm{S}}+n_{\rm{D}}}{n_{\rm{Tot}}}>\frac{ s}{m}.
\end{equation}
However, when \eqref{cond_decreasetime} holds, as $n_{\rm{D}}/n_{\rm{S}}$ increases, eventually a minimal resistance time is achieved. This represents a worse cast scenario yielding fastest possible resistance. For $s\ll d$, this minimal time occurs when
\begin{equation}\label{eqn_mintime}
\frac{n_{\rm{D}}}{n_{\rm{S}}}\approx 1-\frac{2 s  n_{\rm{Tot}}}{ m (n_{\rm{S}}+n_{\rm{D}})}.
\end{equation}

We now examine whether resistance arises via stepwise evolution or direct evolution. Let $T^{\rm{SE}}$ be the first time resistance occurs via stepwise evolution, and $T^{\rm{DE}}$ defined analogously for direct evolution. Then using Theorem \ref{thm_paths}, again with $s\ll d$, we find
\begin{align}\label{eqn_twodrug_prob}
\pr(T^{\rm{SE}}<T^{\rm{DE}})\approx \left[1+\frac{ s (n_{\rm{DD}}+n_{\rm{S}})}{ 
	mn_{\rm{D}} } \right]^{-1}.
\end{align}
Therefore the larger the proportional size of the single drug region, the more likely resistance will arise via stepwise evolution. A qualitatively similar result was derived in \cite{Moreno-Gamez:2015} (e-page 2878) by comparing the most likely stepwise evolution path with the most likely directed evolution path.

Furthermore, in \cite{Moreno-Gamez:2015} a minimal resistance time, in the sense of \eqref{eqn_mintime}, was observed via simulations for the mutation$_A$-migration-mutation$_B$-migration path (the red and dashed path in Fig.\ \ref{fig_twodrugstates}). Using Theorem \ref{thm_paths} we can show that this is the most probable (in the sense of the mode of the path distribution, see \eqref{eqn_pathsthm}) stepwise evolution path whenever $d>s$. Further, by differentiating the median time along this path, we see a minimal time occurs at $n_{\rm{D}}\approx n_{\rm{S}}$, which recovers the simulation result in \cite{Moreno-Gamez:2015} (see their Fig.\ 3). That this path is the most probable stepwise evolution path also gives insight into the recurring $s/m$ term in \eqref{cond_decreasetime} and \eqref{eqn_twodrug_prob}. Relative to direct evolution this path has an extra migration step, so higher migration should \mic{increase the probability of seeding the target via this path}. However fitness costs along this paths are incurred only via the cost of resistance, and a smaller $s$ diminishes these fitness costs. \mic{While the above is motivated by and compared with the study \cite{Moreno-Gamez:2015}, in contrast to that work, we do not consider a size-dependent branching process (the population was constrained by a large carrying capacity in  \cite{Moreno-Gamez:2015}). Despite this we are able to derive similar qualitative, and indeed quantitative (e.g. the recovery of the simulation result discussed above), results. These features may be general properties of such models or are perhaps due to the size of the carrying capacity in \cite{Moreno-Gamez:2015}.}


\subsubsection{Path to resistance in chronic myeloid leukemia}
Let us discuss the example application of \cite{Komarova:2005}, concerning the mechanism by which resistance to the targeted therapy imatinib occurs in chronic myeloid leukemia (CML). The two main resistance mechanisms are via gene amplification of the BCR-ABL fusion gene, which drives the CML, or a point mutation resulting in a modification to the target protein of imatinib. Approximately 100 point mutations have been identified conferring resistance, and a rough estimate for the probability of each of these occurring during a cell division is $10^{-7}$ \cite{Leder:2011}. Thus the probability of a point mutation conferring resistance is approximately $10^{-5}$ per division. The analogous quantity for resistance-causing gene amplifications, in a similar system, was found to be approximately $10^{-4}$ \cite{Tlsty:1989}.  With the birth rate of leukemic stem cells as $\alpha=0.008$ day$^{-1}$ \cite{Michor:2005} we have the rate of point mutations as $\nu_{pm}=8\times10^{-8}$ day$^{-1}$, while the rate of resistance due to gene amplifications is $\nu_{ga}=8\times10^{-7}$ day$^{-1}$. Therefore, using this information alone, the probability of resistance arising from a point mutation first is estimated to be $1/11$. However, in the majority of cases, the primary mechanism of resistance is found to be point mutations \cite{Bixby:2010}.

 Within our framework there are two possible explanations we can explore, both related to the reproductive capabilities of resistant cells. The first is simply that more lineages of amplified cells go extinct. The paragraph directly following \eqref{def_time} indicates how to treat this case. Indeed, if $\rho_{pm}$ is the average survival probability for a population initiated by a single cell with a point mutation and $\rho_{ga}$ the analogous quantity for gene amplification, then we find resistance arising first from a point mutation is more probable if $\rho_{pm}>10\rho_{ga}$. This is possible as point mutations may be deleterious or advantageous \cite{Leder:2011} whereas gene amplifications appear deleterious \cite{Tipping:2001}. 

The second explanation, that suggested in \cite{Komarova:2005}, is that several gene amplifications are required to attain resistance. Suppose a gene amplification increases the death rate of cells by $s$ per day. Then if two gene amplifications are required, ignoring the survival probability from before, the probability that resistance arises via point mutation is now 
$$
\frac{\nu_{pm}}{\nu_{pm}+\nu_{ga}^2/s}=(1+8\times10^{-6}/s)^{-1}.
$$
Resistance via points mutations is now more likely if \mic{$s>10^{-6}$}. Note that Fig.\ \ref{fig_valleyfaster} serves as a schematic for this scenario, reinterpreting the fitness valley path as resistance arising via gene amplifications and the direct path relating to resistance via point mutations. Further explanations for the primacy of point mutation mediated resistance exist and the assumptions above may not hold. But when they do hold, our results offer a clear framework to investigate such questions.


\subsubsection{Antibacterial multidrug resistance}
A timely problem is that posed by the emergence of multidrug resistant bacteria. We now explore some specific examples of our general framework that pertain to this issue. \mic{In particular we consider the accumulation of resistance-conferring mutations as a bacterial colony grows in the absence of drugs}. As before we assume resistance incurs a fitness cost and thus the wild-type \mic{(unmutated cells)} has the highest growth rate. Using Theorem \ref{thm_paths} and measured fitness costs and mutation rates we can predict the most likely path to multidrug resistance.

 We consider the acquisition of resistance to both rifampicin and streptomycin in \textit{P. aeruginosa}. Recently the fitness costs associated with resistance to rifampicin \cite{Hall:2011} and streptomycin \cite{Ward:2009} have been reported. Fitness costs were determined by the maximum growth rate of the bacteria undergoing exponential growth on \mic{nutrient}-rich plates. From \cite{Ward:2009} the ratio of streptomycin resistant bacteria to the wild-type's growth rate was 0.71 (here the growth rate of the wild-type is estimated from Fig.\ 2.a in \cite{Ward:2009}). The same value for resistance to rifampicin, in \cite{Hall:2011}, was 0.88. Let $T^{RS}$ be the time when multidrug resistance emerges assuming rifampicin resistance is acquired first, and $T^{SR}$ defined analogously for streptomycin resistance first. Then from Theorem \ref{thm_paths}, we have that
$$
\pr(T^{RS}<T^{SR})\approx \frac{1-0.71}{(1-0.88)+(1-0.71)}=0.7.
$$
Hence, under this model, we would predict that rifampicin resistance emerges first en route to multidrug resistance. We note that above we have taken the average fitness costs associated with the differing antibiotics. Fluctuations have not been taken into account. 


Instead of considering the the acquisition of resistance to different drugs within a bacterial strain, we can also examine the impact of differing mutation rates across strains. In \cite{Ford:2013} the heightened occurrence of multidrug resistance in lineage 2 vs lineage 4 from \textit{M. tuberculosis} was investigated. Considering in particular resistance to rifampicin and isoniazid, via simulations they deduced that the heightened occurrence of resistance was due to elevated mutation rates in lineage 2 strains. The model simulated can be considered a discrete time counterpart to ours in the case of two paths of length 2 (see Fig.\ \ref{fig_example}). One of the main quantities of interest was the relative probability of resistance, when the only difference between the strains was the mutation rates.  For a particular strain, initiated with one wild-type cell with wild-type birth and growth rates, $\birthone$, $\growthone$, suppose we have mutation rates $\nu_{\rm{R}},\nu_{\rm{I}}$ to rifampicin and isoniazid, each incurring the same fitness cost $s$. Then from Theorem \ref{thm_timesacyclic}, we have
\begin{equation}\label{eqn_probdble}
\pr(T\leq t)\approx 1-\left(\frac{\growthone/\birthone}{1+e^{\growthone(t-\mu)}}+\deathone/\birthone\right),\quad \mu=\frac{1}{\growthone}\log\frac{\growthone^{2}}{\birthone}+\frac{1}{\growthone}\log \frac{s}{2\nu_{\rm{R}}\nu_{\rm{I}}}.
\end{equation}
For lineage-$i$ ($i=2,4$) let $\nu_{\rm{R}}^{(i)}$ and $\nu_{\rm{I}}^{(i)}$ be the mutation rates to rifampicin and isoniazid. As in \cite{Ford:2013}, keeping all other parameters equal between the lineages we have the relative probability of resistance
\begin{equation}
\frac{\pr(T^{(2)}\leq t)}{\pr(T^{(4)}\leq t)}\approx\frac{\nu_{\rm{R}}^{(2)}\nu_{\rm{I}}^{(2)}}{\nu_{\rm{R}}^{(4)}\nu_{\rm{I}}^{(4)}}\frac{(1+2\nu_{\rm{R}}^{(4)}\nu_{\rm{I}}^{(4)}e^{\growthone t}\birthone/\growthone^2 s)}{(1+2\nu_{\rm{R}}^{(2)}\nu_{\rm{I}}^{(2)}e^{\growthone t}\birthone/\growthone^2 s)}.
\end{equation}
For times such that $\nu_{\rm{R}}^{(i)}\nu_{\rm{I}}^{(i)}e^{\growthone t} \ll 1,\,i=2,4$, we see the relative probability is simply $\nu_{\rm{R}}^{(2)}\nu_{\rm{I}}^{(2)}/\nu_{\rm{R}}^{(4)}\nu_{\rm{I}}^{(4)}$. Taking the mutation rates measured in \cite{Ford:2013}, this gives the relative probability of 22.06. This agrees with the simulation results in \cite{Ford:2013} who reported an approximately 22 fold increase. \mic{The same problem was treated in \cite{Colijn:2011} but for multidrug resistance to two drugs in a fixed order (i.e. resistance to rifampicin followed by isoniazid). Staying in the regime $\nu_{\rm{R}}^{(i)}\nu_{\rm{I}}^{(i)}e^{\growthone t} \ll 1$, if in \eqref{eqn_probdble}, we replace the 2 by 1 (as only one path to resistance), and alter both $\growthone^2/\birthone$ by $\growthone$, and $e^{\growthone t}$ with $M$, our expression is the same to leading order as that given in \cite{Colijn:2011} (equation 2) for the probability of resistance in a population of size $M$. This difference between fixed population size and fixed time results is the same as given in \cite{Cheek:2018,Durrett:2010}.}


\subsection{Concluding remarks}

\begin{redsect}
	\subsubsection{Related work}
	Many previous studies have considered the probability of particular cell type emerging in a growing population, typically by a fixed time after the process starts or when the total population reaches a given size. The majority deal with the target population being one or two transitions away from the initial population (the root vertex in this paper) \cite{Luria:1943,Iwasa:2006,Bozic:2013,Colijn:2011,Haeno:2007, Cheek:2018,Antal:2011,Denes:1996}.  In particular we single out the pioneering work of Luria and Delbr\"{u}ck \cite{Luria:1943}, which demonstrated the spontaneous nature of mutations by combining an appropriate mathematical model with bacterial experiments on phage resistance. The original model of \cite{Luria:1943}, and its various incarnations \cite{Lea:1949,Bartlett:1955} have been extensively studied \cite{Mandelbrot:1974,Mohle:2005,Ma:1992,Antal:2011,Kessler:2013,Keller:2015,Cheek:2018}. Its fully stochastic formulation, which is identical to our model for two vertices, is due to Bartlett \cite{Bartlett:1955}. This full model admits an explicit solution \cite{Antal:2011}, and the model's asymptotic behavior has been recently explored \cite{Cheek:2018,Kessler:2013}. One of the simplest quantities of interest is the probability of no resistant cells at a fixed time, often called $p_0$. It is a closely related quantity to our target hitting time described in \eqref{eqn_timeapprox}, that is no resistant cell arises by a fixed time. Understanding $p_0$ provides a method to infer the per cell-rate at which resistance-conferring mutations are acquired, often termed the $p_0$-method \cite{Rosche:2000}.

	Some notable exceptions which consider the target population being greater than two transitions away are \cite{Bozic:2010,Komarova:2005,Bauer:2015,Durrett:2010,Durrett:2011,Moreno-Gamez:2015,Chaumont:2018} (ref. \cite{Moreno-Gamez:2015} is discussed above in Section \ref{sec_twodrugs}). In \cite{Komarova:2005}, the same model as that presented here is numerically explored when all vertices have the same fitness \mic{($\alpha(x)=\alpha$, $\beta(x)=\beta$ for all vertices $x$)}, and the implications on multidrug therapy failure is emphasised. \mic{An efficient numerical method to compute the distribution of target hitting time and path probabilities via the iteration of generating functions is given in \cite{Bauer:2015} with a focus on cancer initiation. Both \cite{Bozic:2010,Durrett:2010} are motivated by the accumulation of driver mutations in cancer, and so each transition leads to a fitness increase. For \cite{Bozic:2010} the mean time of the $k$th driver mutation derived and compared to genetic data. The distribution of the time until the $k$th driver is sought in \cite{Durrett:2010}, whose methods are the closest to those used in this paper. There the authors employed an appealing approximation of the model studied in this paper in the path graph case (the approximation is that the seeding rate into vertex $x+1$ from vertex $x$, which is exactly $\nu(x,x+1)Z_{x}(t)$, is approximated by $\nu(x,x+1) W_x^* e^{\lambda(x) t}$, for some sensibly chosen random variables $W^*_x$). Notably, again for small transition rates, the functional form of the distribution of the target hitting time is the same as that given in Theorem \ref{thm_timesacyclic}, however with a different median. The altered median derived in \cite{Durrett:2010} demonstrates that when transitions lead to a fitness increase, the growth rates associated with the intermediate vertex populations have a far greater effect on the target hitting time when compared to the setting considered in this paper. The same type of approximate model was further used in \cite{Durrett:2011} to investigate the target hitting time when transitions bring a fitness increase which is a random variable. Target hitting times on a path graph for a branching process with general offspring distribution, were also discussed in the very recent paper \cite{Chaumont:2018}, but their main explicit results exclude cell death, and hold when successive transition rates along the path are assumed to be getting (infinitely) stronger.}

	\mic{
		The model and questions considered in this paper also arise frequently in the evolutionary emergence or evolutionary escape literature \cite{Iwasa:2003,Iwasa:2004,Serra:2007,Sagitov:2009,Alexander:2013}, with the notable distinction that in the evolutionary escape setting the root and intermediate vertex populations are destined to go extinct ($\lambda(x)< 0$ for $1\leq x\leq N-1$). This scenario is of interest when, for example, a homogeneous, pathogenic population (all residing at the root vertex in the language used in the present study) is treated with a therapy and must acquire sufficiently many mutations away from the root so as to become resistant (populate the target vertex). As in the setup of the present paper, the target hitting time is strongly controlled by the growth of the root population, $Z_1(t)$, which has positive growth rate, $\growthone>0$, the target hitting times in the evolutionary escape setting are distinct from those given in Section \ref{sec_times}. However for the escape probability (which is the probability of reaching the target), if there are multiple paths from the root to the target, the contribution of each path to the escape probability (termed the path value in  \cite{Iwasa:2003,Iwasa:2004}) has an expression strikingly similar to the path weights discussed here for small transition rates (compare \eqref{def_pathweight} with Eqs 6a-c in \cite{Iwasa:2004}). We might conjecture that for a specific path, the path value, as given in \cite{Iwasa:2004}, is the unnormalised probability of reaching the target via the specified path, as we have demonstrated is the case with the path weights in Theorem \ref{thm_paths} (this is implied in \cite{Iwasa:2004} Sec 2.5). Further connections surrounding the path distribution in these differing regimes is an interesting avenue for future work.
	}

	\subsubsection{Summary}
\end{redsect}

In this article we have considered a continuous time birth-death process with transitions. The time until a population of a particular target state arises was investigated. This target state is accessible via one or more transitions from the initial population. Which sequence of transitions leads to the target population arising was explored. This gave rise to a probability distribution over the paths composed of transitions between differing vertices. Motivated by applications, the setting when the initial population was most fit was focused on. The important factor in determining whether the target was reached via a specific path was seen to be the weight of the path. The weight of the path is composed of the transition rates and fitness costs associated with the intermediate vertices along the considered path.

\mic{Our primary contributions are the simple, explicit formulas for the target hitting time and the distribution of paths to the target, which are valid for small transition rates. When compared with previously given expressions for hitting times in branching processes, due to the reduction in complexity, we believe our results are easily interpretable and widely applicable. Further, to the best of our knowledge, the path distribution presented is the first analytic characterisation of evolutionary paths through a multitype branching process. The utility of our formulas was demonstrated on a variety of scenarios pertaining to the emergence of drug resistance. The biological relations revealed would have been difficult to deduce by a simulation based approach. }

\mic{A shortcoming of our model is that the potential for unlimited growth and the lack of interactions between the cells is unrealistic. However in systems with large carrying capacities and abundant resources, our framework provides a convenient approximation and has been widely used. One possible route for assessing the importance of these neglected factors is to embed our questions in a similar setup to that used in the recent work \cite{Bovier:2018}, which included competition between cell types and a carrying capacity which the population could fluctuate around. The initial conditions in \cite{Bovier:2018} were such that the model does not represent a growing population, so we do not compare their findings with the results of Section \ref{sec_results}. Embedding our questions in a variant of the model of \cite{Bovier:2018}, such that the population is growing, might offer insight on the effect of these neglected aspects. Whenever these factors are deemed unimportant, or as a starting point for more complex models, this article provides quick and accessible results which may be used to guide and develop biological insight.
}


\section{Materials and Methods}\label{sec_methods}

\subsection{Population growth}
The main ingredient in proving Theorem \ref{thm_timesacyclic}, and consequently Theorem \ref{thm_paths}, is the long-time behaviour of the population of the initial and intermediate vertices. In this section we discuss this asymptotic behaviour. 

\mic{
Firstly we extend the definition of the total path weight given in Section \ref{sec_model}. Let the set of paths between the root and vertex $x$ be denoted $\mathcal{P}_{1,x}$. Then using the definition of path weight \eqref{def_pathweight}, we let the total weight for vertex $x$ be $\phi_x=\sum_{p\in \mathcal{P}_{1,x}}w(p)$ for $2\leq x \leq N$. Further, let us denote the ratio of the total weight to fitness cost for vertex $x$ as
\begin{align}
\Phi_{x}=\
\frac{\phi_x}{\lambda-\lambda(x)} 
\end{align}
for $2\leq x \leq N-1$. To allow us to conveniently describe the growth at vertex 1 also, we set $\Phi_1=1$
}

The target population can only be founded by transitions from cells residing at the neighbouring vertices (vertices connected to the target by an edge). Therefore, understanding the population growth of cells at these vertices is needed to discuss the timing and manner in which cells at the target arise. At large times, this understanding is provided by the following theorem, concerning the population at the initial and intermediate vertices \mic{(it is useful to recall that $Z_x(t)$ is the number of cells at vertex $x$ at time $t$ and that we initiate with $z$ cells at vertex 1)}. It follows from a more general result due to Janson \cite{Janson:2004}, which we have tailored to our particular setting. 
\mic{
\begin{theorem}\label{thm_Zk1} With probability one
	\begin{equation}
	\lim_{t\rightarrow \infty}e^{-\growthone t}(Z_x (t))_{x=1}^{N-1}= W (\Phi_x)_{x=1}^{N-1} .
	\end{equation}
	Here $W$ is distributed as the sum of $K$ independent exponential random variables with parameter $\growthone/\birthone$, where $K$ is binomial with $z$ trials and success probability $\growthone/\birthone$.
\end{theorem}
}
\mic{A short proof, demonstrating how Theorem \ref{thm_Zk1} can be deduced from \cite{Janson:2004} is given in Section \ref{sec_SIgrowth} of the supplementary material (see Theorem \ref{thm_Zk2})}.
The distribution of $W$ may be called a binomial-Erlang mixture, being an Erlang distribution with a binomially distributed shape parameter. For the running example of a path graph, we see that asymptotically $Z_{i+1}(t)$ is related to $Z_{i}(t)$ by a multiplicative factor of $\frac{\nu(i,i+1)}{\growthone-\lambda(i+1)}$. This is demonstrated in Fig.\ \ref{fig_growthtime}a. We note that in the case of $N=3$, Theorem \ref{thm_Zk1} was previously given in \cite{Cheek:2018} (Theorem  3.2).

Using Theorem \ref{thm_Zk1} and that the immigration rate into the target vertex from any target neighbouring vertex $x$ is $Z_{x}(t)\nu(x,N)$ we are able to prove the results of Section \ref{sec_results}. Full proofs are provided in the supplementary material.

\section*{Acknowledgments}
We thank David Cheek, Stefano Avanzini and the Edinburgh physics population dynamics group for helpful discussions and feedback. MDN acknowledges support from EPSRC via a studentship.


\section{Supplementary Material}
The outline of this supplement is as follows. First we recap the model, introducing notation that will be needed. Next we prove versions of Theorem \ref{thm_Zk1} and Theorem \ref{thm_timesacyclic}, which do not require $G$ to be acyclic. After this we specialise to the acyclic case. Staying with the acyclic case we prove Theorem \ref{thm_paths}. Finally, how the path distribution is altered when cycles are permitted is considered.

\subsection{Model recap}\label{sec_SImodel}

We recap our framework. Our model is a specific form of multitype branching process \cite{Athreya:2004} in which each population will evolve according to a Markovian linear birth-death process with transitions. We will assume each population is comprised of cells. As a conceptual framework we associate the multitype branching process with a simple, finite, rooted, directed, graph $G=(V,E)$ containing $N$ vertices ($N=|V|$). Labels of the vertices take values in $\{1,\ldots ,N\}$ and $E$ is a subset of the set of ordered pairs $\{(i,j):i,j\in V, i\neq j\}$. We shall often refer to vertex 1 as the root and vertex $N$ as the target. Each vertex is reachable from the root and the target is reachable from any other vertex. 	We further assume the root is a source vertex and the target vertex is a sink. Letting the set of incoming neighbours for vertex $i$ be $\mathcal{N}^{-}(i)=\{k \in V:(k,i)\in E\}$ and the set of outgoing neighbours $\mathcal{N}^{+}(i)=\{k \in V:(i,k)\in E\}$, the previous assumption is $\mathcal{N}^{-}(1)=\mathcal{N}^{+}(N)=\emptyset$. Each type in the branching process is uniquely mapped to a vertex, and so the number of types is $N$. Hence any cell may be described by its type or the vertex associated to that type.

Take any cell at vertex $x$. We assume this cell divides (replaced by two identical copies) at rate $\alpha(x)$, dies (removed from system) at rate $\beta(x)$ and transitions to a cell at vertex $y$ (replaced by an identical copy and a cell at vertex $y$) at rate $\nu(x,y)$ if edge $(x,y)$ is contained in $E$. The growth rate of cells at vertex $x$ will be denoted $\lambda(x)=\alpha(x)-\beta(x)$. The parameters associated with the vertex 1 population feature prominently and so for convenience we let $\alpha=\alpha(1),\,\beta=\beta(1)$ and $\lambda=\lambda(1)$. All cells are independent. We will focus on the setting where the vertex 1 population is the most fit and has positive growth rate. Therefore, we henceforth assume that $\growthone>0$ and for $2\leq x \leq N-1$, $\lambda(x) < \growthone$. We do not specify the relative fitness of the vertex $N$ population. The cell level dynamics may summarised as 
\begin{align*}(x) \rightarrow 
\begin{cases}
(x),(x) \quad &\mbox{rate } \alpha(x)
\\
\varnothing \quad &\mbox{rate } \beta(x)
\\
(x),(y) \quad &\mbox{rate } \nu(x,y) \mbox{ if }  (x,y)\in E
\end{cases}
\end{align*}
where $(x)$ represents a cell at vertex $x$ and $\varnothing$ symbolises a dead cell. At a population level, the number of cells at vertex $x$ at time $t$ will be denoted $Z_{x}(t)$. We shall always assume $Z_{x}(0)=z\delta_{x,1}$, where $\delta_{x,y}$ is the Kronecker delta function.  The population growth of the initial and intermediate vertices is crucial and so the notation $\mathcal{Z}(t)=(Z_x(t))_{1\leq x\leq N-1}$ will be used.

\begin{redsect}
We also recap some relevant definitions from the main text that will be used repeatedly.
First, the target hitting time is defined to be
\begin{equation}\label{defSI_time}
T=\inf\{t\geq 0:\, Z_{N}(t)>0\}. 
\end{equation}
Now let the set of paths between the root and any vertex $x\in V$ be be denoted $\mathcal{P}_{1,x}$. Then we define the weight of the path $p\in \mathcal{P}_{1,x}$ to be
\begin{equation}\label{defSI_pathweight}
w(p)=\nu(p_1,p_2)\prod_{i=2}^{l}\frac{\nu(p_i,p_{i+1})}{\growthone-\lambda(p_i)},
\end{equation}
 Throughout the empty product is set to 1. Further, we let the total weight of the the vertex $x$  be 
\begin{equation}\label{defSI_totalweight}
\phi_x= \sum_{p\in \mathcal{P}_{1,x}}w(p).
\end{equation}
\end{redsect}

\subsection{Population growth}\label{sec_SIgrowth}

In this section we prove Theorem \ref{thm_Zk1}, but $G$ is not assumed to be acyclic. Specialising to the acyclic case is in Section \ref{sec_SIacyclic}. 

As the root vertex is a source, and hence no type transitions into vertex 1 cells, $Z_1(t)$ follows a linear birth-death process. Its asymptotic behaviour is described by the following classic result \cite{Athreya:2004,Durrett:2015}.
\begin{lemma}\label{lemma_limBD} 
	\begin{equation}
	\lim_{t\rightarrow \infty}Z_{1}(t)e^{-\growthone t} =W \quad a.s.
	\end{equation}
		Here $W$ is distributed as the sum of $K$ independent exponential random variables with parameter $\growthone/\birthone$, where $K$ is binomial with $z$ trials and success probability $\growthone/\birthone$.	
\end{lemma}

$W$, as given in the above theorem, is positive if and only if the vertex 1 population survives (almost surely). It will be convenient to define this event and so we let
\begin{equation}\label{def_surv}
\survone =  \{Z_1 (t)> 0\, \mbox{ for all } t\geq 0\}.
\end{equation}
From the distribution of $W$, \mic{in particular using that $\pr(\survone)=\pr(W>0)=\pr(K>0)$ with $K$ as in Lemma \ref{lemma_limBD}}, we immediately have 
\begin{equation}\label{eqn_survoneprob}
\pr(\survone)=1-(\beta/\alpha)^z.
\end{equation}
To obtain a similar result to Lemma \ref{lemma_limBD} for $\mathcal{Z}(t)$, we recall the  $N-1\times N-1$ matrix $A$  from Section \ref{sec_cyc}, with entries
\begin{equation}a_{ij}=
\begin{cases}
\lambda(i)\quad j=i
\\
\nu(j,i)\quad \mbox{if } (j,i)\in E
\\
0\quad \mbox{otherwise.}
\end{cases}
\end{equation}
Note that the off diagonal elements of the transpose of $A$ are positive whenever the corresponding element of the adjacency matrix of $G$ is one.

We will be using a result from \cite{Janson:2004} which holds when the largest real eigenvalue of $A$, $\lambda^*$, is simple (i.e. has algebraic multiplicity 1) and equal to $\growthone$. Below we give two sufficient conditions for this to be true. The second is motivated by applications when typically all transition rates are small.
\begin{lemma}\label{lem_sufcond}
	Two sufficient conditions such that $\lambda^*=\growthone$ and is simple are: 
	1) $G$ is acyclic.
	2) For all $2\leq i \leq N-1$, $\lambda(i)+\sum_{k\in \mathcal{N}^{-}(i)}\nu(k,i) < \growthone$.
\end{lemma}
\begin{proof}
	1) As $G$ is directed and acyclic, a relabelling of the vertices exists (often called the topological sorting, see  \cite{Cormen:2009} chapter 22.4) such that the underlying adjacency matrix of the relabelled graph is upper triangular. We use this relabelling with vertex 1 as first in the labelling, $\pi$. Thus $\pi$ is a permutation of $\{1,\ldots ,N\}$ with $\pi(1)=1$. The permutation matrix associated with $\pi$ and its transpose are
	\begin{equation}
	P_1= \begin{bmatrix}
	\mathbf{e}_{\pi(1)}^{\mathrm{T}} \\
	\vdots \\
	\mathbf{e}_{\pi(n)}^{\mathrm{T}}
	\end{bmatrix}
	,\quad 
	P_1^{\mathrm{T}}= \begin{bmatrix}
	\mathbf{e}_{\pi(1)}, \cdots, \mathbf{e}_{\pi(N)}
	\end{bmatrix}.
	\end{equation}
	Under $\pi$ we have the new matrix
	\begin{equation}
	A'= P_1 A P_1^{-1}.
	\end{equation}
	Due to the relation between the adjaceny matrix of $G$ and $A$, we see $A'$ is lower triangular and has the $\lambda(i)$ as diagonal elements. Thus $\growthone$ is the largest eigenvalue of $A'$. As $A$ and $A'$ are similar matrices, and hence share eigenvalues, we can conclude the statement.

	2) By considering the left eigenvector $\mathbf{e}_{1}^{\mathrm{T}}$, $\growthone$ is indeed an eigenvalue. We now demonstrate all other eigenvalues have real part smaller than $\growthone$, which implies $\lambda$ is simple. Observe that the assumed condition implies the bound on the row sums
	\begin{equation}\label{eqn_rowbound}
	\sum_{j=1}^{N-1} a_{i,j}<\lambda.
	\end{equation}
	As $A$ is reducible (due to the root being a source vertex) there exists a permutation matrix $P_2$ such that
	\[P_2 A P_2^{-1}=
	\left(
	\begin{array}{ccccc}
	\growthone                                    \\
	R_{2,1}& R_{2,2}             &   & \text{\Large0}\\
	\vdots&               & \ddots               \\
	R_{M,1}& \ldots  &   & R_{M,M}            \\
	\end{array}
	\right)
	\]  
	with each submatrix $R_{i,i}$ square and irreducible and $2\leq M\leq N-1$ (see \cite{Varga:1962} section 2.3). The eigenvalues of $A$ comprise $\growthone$ and the eigenvalues of the $R_{i,i}$. For any of the $R_{i,i}$ let $r^{(i)}_j$ be sum of the $j$th row of $R_{i,i}$. By Gershgorin's Disc Theorem (\cite{Varga:1962} Theorem 1.5) the real part of the eigenvalues of $R_{i,i}$ are bounded by $\max_{j}\{r^{(i)}_j\}$. Further, the bound \eqref{eqn_rowbound} is preserved under permutations and so each $r^{(i)}_j<\lambda$. As this holds for each $R_{i,i}$ the claim is shown.
\end{proof}


Henceforth we assume either of the conditions given in Lemma \ref{lem_sufcond} hold and thus $\lambda^* = \growthone$ and is simple. Due to this assumption, results of \cite{Janson:2004} show that the vertex 1 cells drives the entire population growth. To state this result, we introduce the vector $\psi$ with 
\begin{equation}\label{def_psivec}
\psi_{i}=\alpha(i)+\beta(i)+\sum_{k\in \mathcal{N}^{+}(i)}\nu(i,k).
\end{equation}
Now let $\tilde u,\tilde v$ be the left and right eigenvectors of $A$ corresponding to $\growthone$, scaled such that
\begin{equation}\label{cond_evect_norm}
\psi\cdot \tilde v=1,\quad \tilde u\cdot \tilde v = 1.
\end{equation}
Note, from the structure of $A$, $\tilde u_i>0$ if $i=1$ and 0 otherwise. 
Further let
\begin{equation}\label{def_phi}
\rvec=\tilde u_1 \tilde v
\end{equation}
Then the following result, which is the more general form of Theorem \ref{thm_Zk1}, holds.

\begin{theorem}\label{thm_Zk2} 
	With $W$ as in Lemma \ref{lemma_limBD},
	\begin{equation}
	\lim_{t\rightarrow \infty}e^{-\growthone t}\mathcal{Z}(t)= W (\rvec_x)_{x=1}^{N-1} \quad a.s.
	\end{equation}
\end{theorem}
\begin{proof}
	We demonstrate how to apply the result in \cite{Janson:2004} to our setting. From Theorem 3.1 in \cite{Janson:2004}, with probability one
	\begin{equation}\label{eqn_limitJanson}
	\lim_{t\rightarrow \infty}e^{-\growthone t} \mathcal{Z}(t)=\hat W \tilde v
	\end{equation}
	where $\hat W$ is a non-negative random variable, with as yet unknown distribution. However from Lemma \ref{lemma_limBD} almost surely $e^{-\lambda t}Z_1(t)\rightarrow W$, with the distribution of $W$ given in the lemma. This implies $\hat W=W/\tilde v_1=\tilde u_1 W$ with the last equality following from \eqref{cond_evect_norm} and that only the first entry of $\tilde u$ is positive.
\end{proof}
\mic{The proof of Theorem \ref{thm_Zk2}  indicates our primary reason for taking the root vertex to be a source (no incoming edges). If transitions were permitted into the vertex 1 population, but still $\lambda^*=\lambda$ and is simple, then \eqref{eqn_limitJanson} still applies but we do not know the distribution of $\hat W$. Additionally, regardless of whether the root is a source, if we consider the the more general offspring distribution discussed in Section \ref{sec_genoff}, we similarly would have no relevant knowledge of the distribution of $\hat W$ (the Laplace transform of $\hat W$ is known to satisfy an implicit integral equation \cite{Athreya:2004}, but we cannot see how to use this information here).}


\subsection{Time until target vertex is populated: general case}\label{sec_SItimes}
%
%

We now turn our attention to $T$, the time at which the target population arises, as defined in \eqref{defSI_time}. So that we can control all target seeding transition rates we let
$$
\nu^*=\max_{i\in \mathcal{N}^{-}(N)}\nu(i,N)
$$
i.e. the largest of the transition rates into the target. We further give the definition of $\mu$ in this more general setting as
\begin{equation}\label{def_mugeneral}
\mu=\frac{1}{\growthone}\log\frac{\growthone^2}{\birthone\sum_{i\in \mathcal{N}^{-}(N)}\nu(i,N)\rvec_i}.
\end{equation}
Then the following holds.
\begin{theorem}\label{thm_timesgeneral} 
	Let $W$ be as in Theorem \ref{thm_Zk2}, then with probability one
	\begin{equation}\label{eqn_thmtime1}
	\lim_{\nu^*\rightarrow 0}\pr(\growthone(T-\mu)>t|(\mathcal{Z}(t))_{t\geq 0}))=\exp\left(-e^{t}W\growthone/\birthone\right). 
	\end{equation}
	It follows that
	\begin{equation}
	\lim_{\nu^*\rightarrow 0}\pr(\growthone(T-\mu)>t)=\left(\frac{\growthone/\birthone}{1+e^t}+\deathone/\birthone\right)^z
	\end{equation}
	
\end{theorem}

\begin{proof}
	Let 
	$$
	\sigma = \sum_{i\in \mathcal{N}^{-}(N)}\nu(i,N)\rvec_i.
	$$
	Then, as for given $(\mathcal{Z}(t))_{t\geq 0}$ immigration into vertex $N$ occurs as a non-homogeneous Poisson process,
	\begin{equation}\label{eqn_condtimegen}
	\pr(T-\growthone^{-1}\log(\sigma^{-1})> t|(\mathcal{Z}(t))_{t\geq 0}) =\exp\left(-\int_{0}^{t+\growthone^{-1}\log(\sigma^{-1})}\sum_{i\in \mathcal{N}^{-}(N)}\nu(i,N)Z_{i}(s)\,ds\right).
	\end{equation}
	The variable change $u=s -\growthone^{-1}\log(\sigma^{-1})$ yields
	\begin{equation}\label{eqn_int1}
	\exp\left(-\int_{-\growthone^{-1}\log(\sigma^{-1})}^{t}\sum_{i\in \mathcal{N}^{-}(N)}\nu(i,N)Z_{i}(u+\growthone^{-1}\log(\sigma^{-1}))\,du\right).
	\end{equation}
	Now observe that for any $i\in \mathcal{N}^{-}(N)$, by multiplying by $1=\frac{e^{\growthone(u+\growthone^{-1}\log(\sigma^{-1}))}}{e^{\growthone(u+\growthone^{-1}\log(\sigma^{-1}))}}$,
	\begin{equation}\label{eqn_Zrepgen1}
	\nu(i,N)Z_{i}(u+\growthone^{-1}\log(\sigma^{-1}))=\frac{Z_{i}(u+\growthone^{-1}\log(\sigma^{-1}))}{e^{\growthone(u+\growthone^{-1}\log(\sigma^{-1}))}}\frac{\nu(i,N) e^{\growthone u}}{\sigma}.
	\end{equation}
	Consider the limit of \eqref{eqn_Zrepgen1} as $\nu^*\rightarrow 0$, which results in $\log(\sigma^{-1})\rightarrow \infty$. As the $(Z_i(t))_{i\in \mathcal{N}^{-}(N)}$ are independent of $(\nu(i,N))_{i\in \mathcal{N}^{-}(N)}$, \mic{Theorem \ref{thm_Zk2} informs us that, almost surely,
	\begin{equation}\label{eqn_convsum}
	\lim_{\nu^*\rightarrow 0}\frac{Z_{i}(u+\growthone^{-1}\log(\sigma^{-1}))}{e^{\growthone(u+\growthone^{-1}\log(\sigma^{-1}))}}=We^{\growthone u}\rvec_i
	\end{equation}}
	Upon considering $(\sigma/\nu(i,N))^{-1}$, we see $\nu(i,N)/\sigma<\infty$ for each $i$. Hence taking limits across \eqref{eqn_Zrepgen1} yields, with probability one,
	\begin{align}
	\lim_{\nu^*\rightarrow 0}\sum_{i\in \mathcal{N}^{-}(N)}\frac{Z_{i}(u+\growthone^{-1}\log(\sigma^{-1}))}{e^{\growthone(u+\growthone^{-1}\log(\sigma^{-1}))}}\frac{\nu(i,N) e^{\growthone u}}{\sigma}
	&=We^{\growthone u}\lim_{\nu^*\rightarrow 0}\sum_{i\in \mathcal{N}^{-}(N)} \frac{\nu(i,N)\rvec_i}{\sigma}.
	\end{align}
	By the definition of $\sigma$ this last sum is one and hence
	\begin{equation}\label{eqn_vZlim}
	\lim_{\nu^*\rightarrow 0}\sum_{i\in \mathcal{N}^{-}(N)}\nu(i,N)Z_{i}(u+\growthone^{-1}\log(\sigma^{-1}))=We^{\growthone u} \quad a.s.
	\end{equation}
	We seek to apply the dominated convergence theorem to the integral in \eqref{eqn_int1}. \mic{The limit demonstrated in \eqref{eqn_vZlim}} implies that for any realisation, we may find $x$ such that for $\nu^* \leq x$
	\begin{equation}\label{eqn_Zbnd}
	\sum_{i\in \mathcal{N}^{-}(N)}\nu(i,N)Z_{i}(u+\growthone^{-1}\log(\sigma^{-1}))\leq 2 We^{\growthone u}
	\end{equation}
	which is integrable over $(-\infty,t]$. If $\nu^*>x$ then $\sigma>0$. Thus the interval $[-\lambda^{-1}\log(\sigma^{-1}),t]$ is finite. As each $Z_i(t)$ is cadlag, it is bounded over finite intervals. Hence for all $(\nu^*,u)$ the integrand is dominated by an integrable function. Therefore, taking the limit of \eqref{eqn_int1} and using \eqref{eqn_vZlim} leads to, with probability one,
	\begin{align}
	\lim_{\nu^* \rightarrow 0}	\pr(T-\growthone^{-1}\log(\sigma^{-1})> t|(\mathcal{Z}(t))_{t\geq 0})&= \exp\left(-W e^{\growthone t}/\growthone\right).
	\end{align}
	Using that
	$$
	\mu = \growthone^{-1}\log(\growthone^2/\birthone)+\growthone^{-1}\log(\sigma^{-1})
	$$
	so that we centre appropriately, leads to the first stated result.
	
	For the second statement, we must derive an expression for 
	$$
	\ep\left[	\exp\left(-e^{t}W\growthone/\birthone\right)\right]
	$$
	From Lemma \ref{lemma_limBD}
	\begin{equation}
	W\growthone/\birthone\equald\sum_{i=1}^{z}\chi_{i} \xi_{i}'
	\end{equation}
	where $\chi_{i}\sim \mbox{Bernoulli}(\growthone/\birthone)$, $\xi'_{i}\sim \mbox{Exponential}(1)$ and are independent.  Thus 
	\begin{align}\label{eqn_binom}
	\ep\left[	\exp\left(-e^{t}W\growthone/\birthone\right)\right]=\left(\ep\left[\exp\left(-e^{t}\chi_{1}\xi'_1\right)\right]\right)^z.
	\end{align}	
	Using a one-step conditioning argument on $\chi_{1}$ and considering the Laplace transform for a unit rate exponential variable leads to the stated result.
\end{proof}

\mic{Before proceeding, we comment on a limitation of our proof approach. A key point in the proof of Theorem \ref{thm_timesgeneral} is that the stochastic process $(Z_i(t))_{1\leq i\leq  N-1,t\geq 0}$ is independent of the target seeding transition rates $(\nu(i,N))_{i\in \mathcal{N}^{-}(N)}$. This permits the limit displayed in \eqref{eqn_convsum} to hold. If we considered a limit where transition rates associated with non-target adjacent edges also tended to 0, the limit \eqref{eqn_convsum} would not hold. This is as Theorem \ref{thm_Zk2} is true for fixed $(\alpha(x))_{x=1}^{N-1},\ (\beta(x))_{x=1}^{N-1}$ and transition rates associated with non-target adjacent edges (the parameters controlling the growth of the populations at vertices $1,\ldots, N-1)$. Our inability to alter these parameters is also the reason, that in a precise sense, our results do not cover the formulation of transitions where type $x$ divide at rate $\alpha'(x)$, and then with probability $\nu'(x,y)$ a transition occurs to vertex $y$. In that formulation, taking the limit when $\nu'(x,N)$ tends to 0, for $x\in \mathcal{N}^{-}(N)$, alters the division rate of cells at vertices in $\mathcal{N}^{-}(N)$, again forbidding the use of Theorem \ref{thm_Zk2}. We hope to improve upon these limitations in future work. }

The first statement in Theorem \eqref{thm_timesgeneral} will be key to understanding the path distribution, discussed in Section \ref{sec_results_path}, which justifies its prominence. \mic{As it implies that, with probability one,
$$
\lim_{\nu^*\rightarrow 0}\pr(\exp\left[\growthone(T-\mu)\right]>t|(\mathcal{Z}(t))_{t\geq 0}))=\exp\left(-t W\growthone/\birthone\right), 
$$
the first statement in Theorem \eqref{thm_timesgeneral}  can be interpreted as indicating that  if we consider the conditional distribution of $T$ around $\mu$, then rescale time, an exponential distribution appears.}

As discussed in Section \ref{sec_times}, it is natural to consider the distribution of $T$ conditioned that it is finite (we cannot discuss which path populated the target vertex if $T=\infty$). However it is technically more convenient to condition on $\survone$, defined in \eqref{def_surv}. The following proposition states that in many relevant cases, namely large initial population, low death rate or small transition rates leaving vertex 1. these events are similar. Below, the superscript $c$ denotes the complement of the set.

\begin{prop}\label{prop_setdif}
	\begin{align}
	&\pr(\{T<\infty\}^c \cap \survone)=0.\\
	&\pr(\{T<\infty\}\cap S_1^{c})\leq
	\begin{cases}
	O((\deathone/\birthone)^z), \quad &\deathone \rightarrow 0
	\\
	O((\deathone/\birthone)^z), \quad &z\rightarrow \infty
	\\
	O(\nu_1), \quad &\nu_1 \rightarrow 0
	\end{cases}
	\end{align}
	where $\nu_1 = \sum_{x\in \mathcal{N}^{+}(1)}\nu(1,x)$.

\end{prop}
\begin{proof}As $\pr(\{T<\infty\}^c \cap \survone)=\pr(T=\infty|\survone)\pr(\survone)$ and $\pr(\survone)>0$ we initially demonstrate 
	\begin{equation}\label{eqn_TandOmega}
	\pr(T<\infty|S_1)=1.
	\end{equation}
	Firstly recall that with $W$ as in Lemma \ref{lemma_limBD}, up to null events,
	$
	\{W>0\} = S_1.
	$
	Now $T$ is the first arrival in a Cox process with intensity $\sum_{x \in \mathcal{N}^{-}(N)}\nu(x,N)Z_{x}(s)$, and hence $T<\infty$ with probability 1 on $S_1$ if 
	\begin{equation}\label{eqn_intdiverg}
	\int_{0}^{\infty}\sum_{x \in \mathcal{N}^{-}(N)}\nu_{x,N}Z_{x}(s)ds=\infty 
	\end{equation}
	with probability 1 on $S_1$. Observe that \eqref{eqn_intdiverg} is true as, by Theorem \ref{thm_Zk2}, for any realisation we may find $t_1\geq 0$ such that
	$$
	\int_{t_1}^{\infty}\sum_{x \in \mathcal{N}^{-}(N)}\nu(x,N)Z_{x}(s)ds\geq 2^{-1}\int_{t_1}^{\infty}\sum_{x \in \mathcal{N}^{-}(N)}\nu(x,N)W \rvec_x e^{\growthone s} ds.
	$$
	As each $\rvec_x>0$, due to each entry of $\tilde v$ also being positive \cite{Janson:2004}, the right hand side of the above inequality is infinite a.s. on $S_1$ which gives \eqref{eqn_TandOmega}.  It remains to bound $\pr(\{T<\infty\}\cap S_1^{c})$.
	
	Let us introduce $Z^* (t) = \sum_{i\in \mathcal{N}^{+}(1)}Z_{i}(t)$ and 
	$$
	\tau = \inf\{t\geq 0 : Z^*(t)>0\}.
	$$
	Note that we cannot hope to populate the target vertex without $Z^* (t)$ becoming positive. Therefore, $\{T<\infty\}\subseteq \{\tau<\infty\}$ which implies
	\begin{equation}\label{eqn_Tupper}
	\pr(\{T<\infty\}\cap S_1^{c})\leq \pr(\{\tau<\infty\}\cap S_1^{c}).
	\end{equation}
	Hence our interest turns to the right hand side of the preceding inequality. The distribution of $\tau$ is invariant to the growth at vertices $i\in \mathcal{N}^{+}(1)$, and so we let $\alpha(i)=\beta(i)=0$ for such $i$. The process $(Z_1(t),Z^*(t))$ is a two-type branching process where cells of the second type (that contribute to $Z^*(t)$) simply accumulate. As
	\begin{align}\label{eqn_symdif2}
	\pr(\{\tau<\infty\}\cap S_1^{c})&=(1-\pr(\{\tau=\infty\}| S_1^{c}))\pr( S_1^{c})\\
	&=\pr( S_1^{c})-\lim_{t\rightarrow \infty}\pr(Z_1(t)=0,Z^*(t)=0),
	\end{align}
	and from \eqref{eqn_survoneprob} we know $\pr(S_1^{c})=\beta/\alpha$, we now solve for the joint distribution of $(Z_1(t),Z^*(t))$.
	
	\mic{For now assume $z=1$}. We introduce the generating functions
	\begin{align}
	G_1(x,y,t)&=\ep[x^{Z_1(t)}y^{Z^*(t)}|(Z_1(0),Z^*(0))=(1,0)]
	\\
	G^*(x,y,t)&=\ep[x^{Z_1(t)}y^{Z^*(t)}|(Z_1(0),Z^*(0))=(0,1)].
	\end{align}
	Immediately $G^*(x,y,t)=y$ for all $t$. The backward Kolmogorov equation for $G_1$ is
	\begin{align}
	\frac{\partial}{\partial t} G_1 = \birthone G_1^2 +\deathone+\nu_1 G_1 G^*-(\birthone+\deathone+\nu_1)G_1.
	\end{align}
	By solving this equation (separation of variables with a partial fraction decomposition), and using the initial condition we find
	$$
	\pr(Z_1(t)=0,Z^*(t)=0)=G_1(0,0,t)=\frac{r_2(1-e^{\mic{\birthone}(r_2-r_1)t})}{1-r_2 e^{\mic{\birthone}(r_2-r_1)t}/r_1}.
	$$
	With 
	\begin{align*}
	r_{1}=\frac{1}{2}\left(1+q_{1}+\nu_1'+\sqrt{(1+q_{1}+\nu_1')^2-4 q_{1}  }\right)
	\\
	r_{2}=\frac{1}{2}\left(1+q_{1}+\nu_1'-\sqrt{(1+q_{1}+\nu_1')^2-4 q_{1}  }\right)
	\end{align*}
	where $q_{1}=\deathone/\birthone$ and $\nu_1'=\nu_1/\birthone$. Observe that $(1+q_{1}+\nu_1')^2-4 q_{1}>(1-q_1)^2$ which implies $r_2-r_1<0$. Hence
	$$
	\lim_{t\rightarrow \infty}\pr(Z_1(t)=0,Z^*(t)=0)=r_2.
	$$
	We now have expressions for the final terms in \eqref{eqn_symdif2} for $z=1$. For $z\geq 1$, we can use independence, and this leads to 
	$$
	\pr(\{\tau<\infty\}\cap S_1^{c})=q_1^z-r_2^z.
	$$

	By examining the leading order with respect to the relevant variables we can conclude.

\end{proof}

We remark that if we instead wish to look at the case when the vertex $N$ population arises and does not go extinct then $\pr(\mbox{All populations go extinct})=\pr(\tilde T=\infty)$ \mic{where $\tilde T$ is the first time the target is founded by a cell whose progeny does not go extinct. Recall that our results for $T$ hold identically for $\tilde T$ so long as the mapping $\nu(x,N)\mapsto \nu(x,N)\lambda(x)/\alpha(x)$ for all $x\in \mathcal{N}^{-}(N)$ is applied} . The probability of population extinction is known to be given by an implicit equation involving the offspring distribution's generating function \cite{Athreya:2004}. This was the approach taken by \cite{Iwasa:2004}. The approximate solution for small transition rates was derived and agrees with the result presented above \mic{(see Eq. A.4a in Appendix A of \cite{Iwasa:2004})}.

The above proposition yields $\pr(T<\infty)\approx\pr(\survone)$. Using this with \eqref{eqn_survoneprob} gives \eqref{eqn_proboccurs} in the main text. We now give consider the distribution of $T$ given $S_1$ (which we use as proxy for conditioning on a finite hitting time).

\begin{corollary}\label{cor_condtime} With the notation of Theorem \ref{thm_timesgeneral},
	\begin{equation}
	\lim_{\nu^*\rightarrow 0}\pr(\growthone(T-\mu)>t|\survone)=\frac{\left[\growthone/\birthone (1+e^t)^{-1}+\deathone/\birthone\right]^z-(\deathone/\birthone)^z}{1-(\deathone/\birthone)^z}
	\end{equation}
\end{corollary}
\begin{proof}
	Starting from the first statement in Theorem \ref{thm_timesgeneral}, we only need apply $\ep[\cdot|S_1]$.  The distribution of $W\growthone/\birthone$ can be rewritten as
	$$
	W\growthone/\birthone\equald  \sum_{i=1}^{K}\xi'_i
	$$
	where $\xi'_i$ are rate one exponential variables, and $K\sim \mbox{Binomial}(z,\growthone/\birthone)$, independent of all $\xi'_i$. Conditioning on $\survone$ ensures that $K>0$.

	Thus we must calculate\mic{
	\begin{align}
	\ep\left[	\exp\left(-e^{t}W\growthone/\birthone\right)|\survone\right]=\ep\left[(1+e^t)^{-K}|K>0\right].
	\end{align}}
	The binomial theorem leads to the stated expression.
\end{proof}

%
%

We now prove a corollary regarding $\mu$ (\mic{defined in \eqref{def_mugeneral}}), which provides justification to \eqref{eqn_medtime}.

\begin{corollary}\label{cor_medpath}
	Let $t_{1/2}$ be the median time for the vertex N population to arise, conditioned on $\survone$. That is $t_{1/2}$ satisfies
	\begin{equation}\label{eqn_med1}
	\pr(T >t_{1/2}|S_1)=\frac{1}{2}.
	\end{equation}
	Then 
	\begin{equation}
	\lim_{\nu^*\rightarrow 0}|t_{1/2}-(\mu-h(z))|=0
	\end{equation}
	where 
	\begin{equation}
	h(z)=-\growthone^{-1} \log\left(\frac{\growthone/\birthone}{2^{-1/z}(1+(\deathone/\birthone)^{z})^{1/z}-\deathone/\birthone}-1 \right).
	\end{equation}
	This also implies $t_{1/2}\sim \mu-h(z)$. Regarding $h(z)$ we have
	\begin{align}
	h(1)=0,\quad h(z) &= \growthone^{-1}\log\left(\frac{z \growthone}{\birthone }\right)+O(z)^{-1},\quad z \rightarrow \infty.
	\end{align}
\end{corollary}

\begin{proof}
	We firstly note, as is apparent from the cumulative distribution function of $T$ (see \eqref{eqn_condtimegen} for instance), that $t_{1/2}$ is monotone increasing as $\nu^*$ decreases. By \eqref{eqn_med1}
	\begin{equation}
	\frac{1}{2}=\pr(\growthone(T-\mu)>\growthone(t_{1/2}-\mu)|\survone).
	\end{equation}
	Taking the $\nu^* \rightarrow 0$ limit we find $\growthone(t_{1/2}-\mu)\rightarrow x$, where the convergence is guaranteed as both $t_{1/2},\,\mu$ are monotone increasing. Further, from Corollary \ref{cor_condtime}, the limit $x$ satisfies
	\begin{equation}
	\frac{[\growthone/\birthone (1+e^x)^{-1}+\deathone/\birthone]^z-(\deathone/\birthone)^z}{1-(\deathone/\birthone)^z}=\frac{1}{2}.
	\end{equation}
	Solving for $x$ leads to $x=-\growthone h(z)$ as stated. For the asymptotic result we firstly consider 
	\begin{align}
	d(z) &=\frac{\growthone/\birthone}{2^{-1/z}(1+(\deathone/\birthone)^{z})^{1/z}-\deathone/\birthone}-2^{\frac{\birthone}{\growthone z}}
	\end{align}
	Expansions on the numerator and denominator show that $d(z)$ is $O(z^{-2})$. Then 
	\begin{align}
	-\growthone h(z) &= \log(2^{\frac{\birthone}{\growthone z}}-1+d(z))
	\\
	&=\log(2^{\frac{\birthone}{\growthone z}}-1)+\frac{d(z)}{2^{\frac{\birthone}{\growthone z}}-1}+\ldots
	\end{align}
	As $(2^{\frac{\birthone}{\growthone z}}-1)^{-1}=O(z)$ we have shown the claimed asymptotic behaviour.
	
\end{proof}

\subsection{Path graph and acyclic graphs}\label{sec_SIacyclic}
In this section we specialise the results of Sections \ref{sec_SIgrowth} and \ref{sec_SItimes} firstly to the case of a path graph and then to acyclic graphs.

We first consider a path graph with $N$ vertices and let $\tilde p_{1:x}=(1,2,\ldots x)$. 
In this setting the matrix $A$ from Section \ref{sec_SIgrowth} has entries
\mic{
\begin{equation}a_{ij}=
\begin{cases}
\lambda(i)\quad j=i
\\
\nu(j,i)\quad j=i-1
\\
0\quad \mbox{otherwise},
\end{cases}
\end{equation}
}
and the vector $\psi$ has entries $\psi_{i}=\alpha(i)+\beta(i)+\nu(i,i+1)$ for $i=1,\ldots, N-1$. 
We want a useful representation of the sequence $(\rvec_x)_{x=1}^{N-1}$, which we recall is given in terms of the entries of the suitably normalised left and right eigenvectors of $A$ (see \eqref{cond_evect_norm}). Because they are easier to work with (i.e. check they indeed are eigenvectors) we first introduce the unnormalised left and right vectors $u',\,v'$ with entries\tr{
\begin{equation}
u'_{i} =\begin{cases}
1,\quad &i=1
\\
0 \quad & 2\leq i\leq N-1
\end{cases}
\hspace{1.5cm}
v'_{i}=\prod_{j=2}^{i}\frac{\nu(j-1,j)}{\growthone-\lambda(j)}, \quad  1\leq i \leq N-1
\end{equation}
}
where the empty product is again set equal to 1. \mic{Note for $2\leq x \leq N-1$ we have 
	$$
	v'_{x}=\frac{w(\tilde p_{1:x})}{\growthone-\lambda(x)},
	$$ }
\mic{here $w(\cdot)$ is again the path weight defined in \eqref{defSI_pathweight}} . It may be directly verified that these are left and right eigenvectors for $A$ with eigenvalue $\growthone$ (\mic{for $v'$ it is helpful to note that $A v'=\growthone v'$ implies $\nu(i-1,i)v_{i-1}' +\lambda(i) v'_i= \lambda v'_i$)}  . Now we normalise these, so that they agree with \eqref{cond_evect_norm} via
\begin{equation}
\tilde u=\frac{\psi\cdot v'}{v'_1}u',\quad \tilde v=\frac{1}{\psi\cdot v'}v'.
\end{equation}
\mic{Our aim is to obtain $\Phi = \tilde u_1 \tilde v$, and as both $u'_1=v'_1=1$, we see that $\Phi = v'$   }. Using this, the following corollary demonstrates how the results of the previous section on population growth and time to reach the target simplifies in the path graph setting.

\begin{corollary}\label{cor_pathgraph}
	The results of Sections \ref{sec_SIgrowth} and \ref{sec_SItimes} hold with the explicit representations \mic{
	\begin{align}
 \rvec_x=\begin{cases}
1 \quad &x=1
\\
\frac{w(\tilde p_{1:x})}{\lambda-\lambda(x)} \quad &2\leq x \leq N-1
 \end{cases}
	\end{align}}
	$$
	\nu^* = \nu(N-1,N),
	$$
	and
	\begin{equation}
	\mu=\growthone^{-1}\log\left[\growthone^2\left(\birthone \phi_N\right)^{-1}\right].
	\end{equation}
\end{corollary}
\mic{The appearance of $\phi_{N}$ in the expression for $\mu$ in Corollary \eqref{cor_pathgraph} is due to $\phi_{N} =w(\tilde p_{1:N})= \nu(N-1,N)\rvec_{N-1}$} \mic{(to be compared with the general definition of $\mu$ \eqref{def_mugeneral}). For referencing, note the definition for the weight of a path $p$, $w(p)$, and total weight of $N$, $\phi_N$, are in \eqref{defSI_pathweight} and \eqref{defSI_totalweight}.} We now move to the case where $G$ is acyclic.


Recall the concept of vertex lineage from Section \ref{sec_results_path} in the main text \mic{(which for a chosen cell, is a sequence recording the vertices of cells in the ancestral lineage of our chosen cell)}. We denote the number of cells at time $t$ with vertex lineage $p$ as $X(p,t)$. The number of cells at a vertex is related to the cells with vertex lineages ending at that vertex via
\begin{equation}\label{eqn_ZfromXgen}
Z_x (t) = \sum_{p \in \mathcal{P}_{1,x}} X(p,t).
\end{equation}


For any path $p$, as above  we let the first $j$ terms in $p$ be denoted $p_{1:j}$. The process $(X(p_{1:j},t))_{j=1}^{\pl+1}$ is a multitype branching process, and $X(p_{1:j},t)$ represents the number of cells that have progressed to step $j$ along path $p$ at time $t$. Births occur to individuals of type $i$ at at rate $\alpha(p_i)$, deaths at rate $\beta(p_i)$ and individuals of type $i$ transition to type $i+1$ at rate $\nu(p_i,p_{i+1})$. Thus we may use our results for path graphs with vertices labelled $1,\ldots,\pl+1$. Hence via Theorem \ref{thm_Zk2} and Corollary \ref{cor_pathgraph} we have the following. 
\begin{corollary}\label{cor_Xlim}
	\begin{equation}
	\lim_{t\rightarrow \infty}e^{-\growthone t}X(p,t)=  W w(p) \quad a.s.
	\end{equation}
	where
	$w(p)$ is the path weight given in \eqref{defSI_pathweight} and $W$ is as in Lemma \ref{lemma_limBD}. 
\end{corollary}
Analogously to Corollary \ref{cor_pathgraph} we have the following statement regarding the growth of the population at the initial and intermediate vertices.
\begin{corollary}\label{cor_acyclicgrowth}
	The results of Sections \ref{sec_SIgrowth} and \ref{sec_SItimes} hold with the explicit representations
	\mic{
		\begin{align}
	\rvec_x=\begin{cases}
	1 \quad &x=1
	\\
	\frac{1}{\lambda-\lambda(x)} 	\sum_{p\in\mathcal{P}_{1,x}}w(p)\quad &2\leq x \leq N-1
	\end{cases}
	\end{align}}
		\begin{equation}
	\mu= \growthone^{-1}\log\left(\growthone^2\left[\birthone  \sum_{p\in\mathcal{P}_{1,N}}w(p)\right]^{-1}\right).
	\end{equation}
	where	$w(p)$ is the path weight given in \eqref{defSI_pathweight}.
\end{corollary}
Note that as the empty product is defined to be 1, $\phi_1=1$.
\begin{proof}[Proof of Corollary \ref{cor_acyclicgrowth}]
	Due to Theorem \ref{thm_Zk2} we must demonstrate
	\begin{equation}
	\lim_{t\rightarrow \infty}e^{-\growthone t}Z_x (t)= W \sum_{p\in\mathcal{P}_{1,x}}w(p) \quad a.s.
	\end{equation}
	By \eqref{eqn_ZfromXgen} it is enough to show
	\begin{equation}
	\pr\left(\bigcap _{p\in\mathcal{P}_{1,x}} \lim_{t\rightarrow \infty} e^{-\growthone t}X(p,t)=W\right)=1.
	\end{equation}
	Observe that for any events $(A_i)_{i=1}^{n}$ such that for all $i$, $\pr(A_i)=1$ then $\pr(\cap_{i}^{n}A_i)=1$. Use this with events $A_p = \{ \lim_{t\rightarrow \infty} e^{-\growthone t}X(p,t)=W\}$ to conclude. 
\end{proof}

One particular consequence, by coupling Corollaries \ref{cor_condtime} and \ref{cor_acyclicgrowth} is that for $\nu^*$ small
\begin{equation}\label{eqn_timeapprox_cond}
\pr(T>t|S_1)\approx \frac{\left[\growthone/\birthone (1+e^{\growthone t}\phi_{N}\birthone/\growthone^2)^{-1}+\deathone/\birthone\right]^z-(\deathone/\birthone)^z}{1-(\deathone/\birthone)^z}
\end{equation}
which is the conditional version of \eqref{eqn_timeapprox}. Differentiating yields an approximation for the conditional density
\begin{equation}\label{eqn_pdfcond}
f_{T|T<\infty}(t)\approx \frac{z e^{\lambda t}\phi_N\left(\frac{\growthone/\birthone}{1+e^{\lambda t}\phi_N \alpha/\lambda^2}+\deathone/\birthone\right)^{z-1}}{\left(1-(\beta/\alpha)^z\right)\left(1+\frac{\alpha e^{\lambda t}\phi_N}{\lambda^2}\right)^2}.
\end{equation}

\subsection{Distribution of the path to the target}\label{sec_SI_pathproof}
We continue to take $G$ acyclic. Note that $|\mathcal{P}_{1,N}|$ is finite, and so we let 
$
n=|\mathcal{P}_{1,N}|.
$
The elements of $\mathcal{P}_{1,N}$ can be enumerated and we use the following notation to do so
$
\mathcal{P}_{1,N}=\{p^{(1)},\ldots, p^{(n)}\}.
$
For each $p^{(i)}\in\mathcal{P}_{1,N}$ we let: the path lengths be $\pl^{(i)}=|p^{(i)}|-1$, $u^{(i)} = \nu(p^{(i)}_{\pl^{(i)}},N)$ be the final transition rates, and $\eta^{(i)}=w(p^{(i)})/u^{(i)}$. Consider $p^{(i)}\in \mathcal{P}_{1,N}$ and recall that $X(p^{(i)}_{1:j},t)$ represents the number of cells that have progressed to step $j$ along path $i$ at time $t$. If we apply the first statement in Theorem \ref{thm_timesgeneral} to the branching process $(X(p^{(i)}_{1:j},t)))_{j=1}^{\pl^{(i)}+1}$ we can deduce that under a time rescaling, the time to traverse path $i$ is approximately exponentially distributed, conditional on the population growth along that path. Therefore the question of which path the target population arises from becomes equivalent to the minimum of a set of exponential random variables. To simplify notation we let,
\begin{equation}
T^{(i)} =T(p^{(i)})
\end{equation}
where $T(p^{(i)})$ is defined in \eqref{def_pathtime}. \mic{Also recall the notation 
	\begin{equation}\label{def_SInotptime}
	T(\neg p)=\min\{T(q):q\in \mathcal{P}_{1,N}\backslash\{p\}\}.
	\end{equation}}
	We can now prove the precise version of Theorem \ref{thm_paths}
\begin{theorem}\label{thm_paths2} 
 Assume the limits 
$
\lim_{\nu^*\rightarrow 0}\frac{u^{(i)}}{u^{(j)}}
$
exist, but may be infinite, and further that there exists $i^*\in \{1,\ldots ,n\}$ such that 
$
\lim_{\nu^*\rightarrow 0}\frac{u^{(i)}}{u^{(i^*)}}<\infty
$ for all $i=1,\ldots n$.  Then, for $p\in \mathcal{P}_{1,N}$ and $t\in \mathbb{R}$,
	\begin{equation}
\lim_{\nu^*\rightarrow 0}\pr(T(\neg p)-T(p)>t|S_1)=\lim_{\nu^*\rightarrow 0} \frac{w(p)}{w(p)+e^{\growthone t}\sum_{q\in \mathcal{P}_{1,N}\backslash\{p\}}w(q)}.
\end{equation}
\end{theorem}

\begin{proof}
	Consider initially the case when $n=2$, that is there are only two paths from vertex 1 to vertex $N$. These two paths must diverge before $N$. If the population growth along these paths until the target vertex \mic{(but excluding the population growth at the target vertex itself)} is given, then $T^{(1)}$ and $T^{(2)}$ are the first arrival times in two independent non-homogeneous Poisson processes.  Returning to arbitrary $n$, to use this fact we define the population numbers along all the paths up until the target vertex as
	$$
	\mathcal{X}(s)=(X(p^{(i)}_{1:j},s))_{i=1,j=1}^{n,\pl^{(i)}}.
	$$
	So that we can use this conditional independence, and prior results, we note
	\begin{equation}\label{eqn_towerrule}
	\pr(T(\neg p)-T(p)>t|S_1)=
	\frac{\ep[\mathbb{I}_{S_1}\pr(T(\neg p)-T(p)>t|(\mathcal{X}(s))_{s\geq 0})]}{\pr(S_1)}.
	\end{equation}
	We now focus on the argument of the above expectation. Let the path under consideration be path 1 ($p=p^{(1)}$) and
	\begin{equation}\label{def_kappa}
	\kappa^{(i)}=\lim_{\nu^*\rightarrow 0}\frac{u^{(i)}}{u^{(i^*)}}.
	\end{equation}
	We know, \mic{by Theorem \ref{thm_timesgeneral}, } that each $T^{(i)}$ appropriately centred converges. If we centre each $T^{(i)}$ by the same quantity, then the ordering is preserved. That is, clearly
	$
	T^{(1)}<T^{(2)} \iff T^{(1)}-\mu^{(i)}<T^{(2)}-\mu^{(i)}$	for any $\mu^{(i)}$.
	As $T^{(i)}$ is the target time along the path graph $p^{(i)}$, Corollary \ref{cor_pathgraph} gives the  appropriate centring. In particular with
	\begin{equation}
	\mu^{(i)} = 	\growthone^{-1}\log(\growthone^2[\birthone\eta^{(i)} u^{(i)}]^{-1})
	\end{equation}
	we have, \mic{by the first statement of Theorem \ref{thm_timesgeneral} and Corollary \ref{cor_pathgraph}},
	\begin{equation}
	\lim_{\nu^*\rightarrow 0}\pr(\growthone(T^{(i)}-\mu^{(i)})>x|(\mathcal{X}(s))_{s\geq 0}) =\exp(-e^x W \growthone/\birthone)\, a.s.
	\end{equation}
	For order preservation (of the $T^{(i)}$s) we take $\mu^{(i^*)}$ as a common centring of all $T^{(i)}$, and use $T^{(i)}-\mu^{(i^*)}=T^{(i)}-\mu^{(i)} + \mu^{(i)} -\mu^{(i^{*})}$ and 
	\begin{equation}
	\lim_{\nu^*\rightarrow 0}\mu^{(i)}-\mu^{(i^*)} = \growthone^{-1} \log\left(\frac{\eta^{(i^*)}}{\eta^{(i)} \kappa^{(i)}} \right).
	\end{equation}
	Note the above expression may be $-\infty$. Using this, and exponentiating (which again preserves order), leads to
	\begin{equation}
	\lim_{\nu^*\rightarrow 0}\pr(\exp[\growthone(T^{(i)}-\mu^{(i^*)})]>t|(\mathcal{X}(s))_{s\geq 0}) =\exp\left(-\frac{t W\growthone \eta^{(i)}\kappa^{(i)}}{\birthone\eta^{(i^*)}}\right)\, a.s.
	\end{equation}
	At this point we note that if $\kappa^{(i)}=0$, then for fixed $(\mathcal{X}(s))_{s\geq 0}$, the random variables $\exp[\growthone(T^{(i)}-\mu^{(i^*)})]$ converges to $+\infty$ almost surely. Recall the $T^{(i)}$ are conditionally independent, and hence convergence of the marginal distributions implies joint convergence. Therefore for $(x_{j})_{j=1}^{n}$ with each $x_j\geq 0$, almost surely we have
	\begin{equation}
	\lim_{\nu^*\rightarrow 0}\pr\left[\left(\exp[\growthone(T^{(j)}-\mu^{(i^*)})]>x_{j}\right)_{j=1}^{n}|(\mathcal{X}(s))_{s\geq 0}\right]=\prod_{j=1}^{n}\exp \left( -\frac{x_{j}W \kappa^{(j)}\eta^{(j)}\growthone}{\eta^{(i^*)}\birthone}\right).
	\end{equation}
	Including the shift $t$, we have for fixed $(\mathcal{X}(s))_{s\geq 0}$ the random vector $(\exp[\growthone(T^{(1)}+t-\mu^{(i^*)})],\exp[\growthone(T^{(j)}-\mu^{(i^*)})])_{j=2}^{n}$ jointly tends in distribution to the random vector $(U_i)_{i=1}^{n}$. \mic{Note that only $T^{(1)}$ is shifted, due to the path under consideration being $p^{(1)}$}. If $W=0$, all the $U_i=+\infty$ a.s., instead if $W>0$ each $U_{i}$ is independent and for $2\leq i \leq n$ has distribution 
	\begin{equation}
	U_i \equald \begin{cases}
	+\infty  \mbox{ if } \kappa^{(i)}=0
	\\
	E_i\mbox{ if } 0<\kappa^{(i)}<\infty
	\end{cases}
	\end{equation}
	where \mic{$E_i$ is an exponential random variable with parameter $\frac{W \kappa^{(i)}\eta^{(i)}\growthone}{\eta^{(i^*)}\birthone}$}. For $i=1$, instead we have $U_1\equald +\infty$ if $\kappa^{(1)}=0$, else $U_1$ is distributed as an exponential with parameter $\frac{e^{-\growthone t}W \kappa^{(1)}\eta^{(1)}\growthone}{\eta^{(i^*)}\birthone}$. As, up to null events, $\{W>0\}=S_1$, and using standard results for the ordering of exponential variables (looking at the minimum of the $U_i$), we have
	\mic{
	\begin{align}\label{eqn_condlimt}
	\lim_{\nu^*\rightarrow 0}\mathbb{I}_{S_1}\pr(T(p)+t<T(\neg p)|(\mathcal{X}(s))_{s\geq 0})=\mathbb{I}_{S_1}\frac{\kappa^{(1)}\eta^{(1)}}{\kappa^{(1)}\eta^{(1)}+ e^{\growthone t}\sum_{j=2}^{n}\kappa^{(j)}\eta^{(j)}}.
	\end{align} 
}
	The right hand side of \eqref{eqn_condlimt} is an element of $[0,1]$. $0$ is obtained when $\kappa^{(1)}=0$ (the final transition rate along the path is small, \mic{in the sense of the limit \eqref{def_kappa}}, relative to other paths) and $1$ is obtained if $\kappa^{(i)}=\delta_{1,i}$ (the final transition rate along the path is the largest relative to other paths). 	Coupling the above with \eqref{eqn_towerrule} gives the result.
\end{proof}
%
%


\subsection{Cyclic graphs}\label{sec_cyc2}

Suppose $G$ is as before, i,e. Section \ref{sec_SI_pathproof}, but without the acyclic assumption. For cyclic graphs in place of $\mathcal{P}_{1,N}$ we must consider $\mathcal{W}_{1,N}$, which we define to be the set of walks between the root and vertex $N$. The definition of vertex lineages holds identically for walks. Therefore for $\walk\in \mathcal{W}_{1,N}$ we let
\begin{equation}\label{def_Tw}
T(\walk) = \inf\{t\geq 0: X(\walk,t)>0\}
\end{equation}
where $X(\walk,t)$ is as given in Section \ref{sec_results_path}.  Here we show that so long as the number of transitions between vertices is finite, which corresponds to finitely many back transitions, we may map a graph containing cycles to an acyclic one and then use the results given above.

The set of walks between the root and the target vertex of length at most $i-1$ will be denoted
$$
\mathcal{W}_{1,N}^{(i)}=\{\walk\in \mathcal{W}_{1,N}: |\walk|\leq i\}.
$$
Further for a graph $G=(V,E)$, let the vertex parameters be the birth and death rates associated to the vertices of $G$ and the edge parameters the transition rates associated with the edges.
\begin{prop}\label{prop_acyctocyc}
	For a given graph $G=(V,E)$ with associated vertex and edge parameters, \mic{and $i\geq 1$ which is the upper-length of walks that we consider on $G$}, we can construct an acyclic graph $G'=(V',E')$ (dependent on $i$), which coupled with appropriate vertex and edge parameters, permits a bijection $g:\mathcal{W}_{1,N}^{(i)}\mapsto \mathcal{P}'_{1,|V'|}$. This bijection possesses the property that for $\walk\in \mathcal{W}_{1,N}^{(i)}$
	$$
	T(\walk) \equald T'({g(\walk)})
	$$
	where $T'({g(\walk)})$ is defined as in \eqref{def_Tw} but with the process on $G'$ with its associated parameters.
\end{prop}
Before giving the proof we briefly demonstrate this on the graphs shown in Fig.\ \ref{fig_cyctoacyc}. This elucidates the general idea.

\noindent\begin{minipage}[t]{1\linewidth}

	\begin{center}
		\begin {tikzpicture}[-latex ,auto ,node distance =3.1 cm and 2cm ,on grid ,
		semithick ,
		state/.style ={ circle ,top color =white , bottom color = processblue!20 ,
			draw,processblue , text=blue , minimum width =.11 cm}]
		\node[state,label={[label distance=0.35cm]90:$G$}] (1)  at (0,3){$1$};
		\node[state] (2) [below left of =  1] {$2$};
		\node[state] (3) [below right of =  1] {$3$};
		\node[state] (4) [below right of =  2] {$4$};

		\path (1) edge node[above =0.15 cm] {} (2);
		\path (1) edge node[above =0.15 cm] {} (3);
		\path (2) edge node[above =0.15 cm] {} (3);
		\path (3) edge node[above =0.15 cm] {} (2);
		\path (2) edge node[above =0.15 cm] {} (4);
		\path (3) edge node[above =0.15 cm] {} (4);

	\end{tikzpicture}
	\qquad 
	\begin {tikzpicture}[-latex ,auto ,node distance =2 cm and 2cm ,on grid ,
	semithick ,
	state/.style ={ circle ,top color =white , bottom color = processblue!20 ,
		draw,processblue , text=blue , minimum width =.11 cm}]
	\node[state,label={[label distance=0.35cm]90:$G'$}] (1)  at (0,3){$1$};
	\node[state] (2) [below left of =  1] {$2$};
	\node[state] (3) [below left of =  2] {$3$};
	\node[state] (4) [below right of =  1] {$4$};
	\node[state] (5) [below right of =  4] {$5$};
	\node[state] (6)  at (0,-1.5){$6$};

	\path (1) edge node[above =0.15 cm] {} (2);
	\path (1) edge node[above =0.15 cm] {} (4);
	\path (2) edge node[above =0.15 cm] {} (3);
	\path (2) edge node[above =0.15 cm] {} (6);
	\path (4) edge node[above =0.15 cm] {} (5);
	\path (4) edge node[above =0.15 cm] {} (6);
	\path (3) edge node[above =0.15 cm] {} (6);
	\path (5) edge node[above =0.15 cm] {} (6);
	
\end{tikzpicture}
\captionof{figure}{}
\label{fig_cyctoacyc}
\end{center}
\vspace{.3cm}
\end{minipage}
Here 
$$
\mathcal{W}_{1,N}^{(4)}=\{(1,2,4),(1,2,3,4),(1,3,4),(1,3,2,4)\}.
$$
It is clear these may be mapped to the paths in $G'$
$$
\mathcal{P}'_{1,|V'|}=\{(1,2,6),(1,2,3,6),(1,4,6),(1,4,5,6)\}.
$$
Further $T((1,2,4))\equald T'((1,2,6))$ so long as the transition rates associated with edges $(1,2)$ and $(2,4)$ on $G$ are equal to the rates with $(1,2)$ and $(2,6)$ on $G'$ and that the birth and death rates at vertices $1$ and $2$ are equal in both cases.

\begin{proof}[Proof of Proposition \ref{prop_acyctocyc}]
To avoid cumbersome notation and as the idea is illustrated in Fig.\ \ref{fig_cyctoacyc} we sketch an algorithm. We enumerate the set $\mathcal{W}_{1,N}^{(i)}=\{\walk^{(1)},\walk^{(2)},\ldots \walk^{(\mathcal{M})}\}$, where $\mathcal{M}=|\mathcal{W}_{1,N}^{(i)}|$. Again we use the notation $l^{(j)}$ for the length of the $j$th walk, that is $l^{(j)}=|\walk^{(j)}|-1$. The graph $G'$ is constructed iteratively. Let $V'_1=\{1,2,\ldots, l^{(1)},N'\}$, where $N'$ is a placeholder and will ultimately be $|V'|$. Similarly $E'_1=\{(1,2),\ldots,(l^{(1)},N')\}$ and $p^{(1)}=(1,2,\ldots, l^{(1)},N')$.

Then for each further $\walk^{(j)}$, $2\leq j\leq \mathcal{M}$, from the previously considered walks select \mic{the walk that is identical to $\walk^{(j)}$ for the greatest number of vertices}, say $\walk^{(k)}$, $1\leq k <j$. If $\walk^{(j)}$ and $\walk^{(k)}$ agree up until their $m$th element, and $l^{(j)}-m>0$ then add $l^{(j)}-m$ elements to $V'_{j-1}$ (new vertices starting with $|V'_{j-1}|$). This creates $V'_j$. Denote the new vertices as $(v^{(j)}_x)_{x=1}^{l^{(j)}-m}$. \mic{As $\walk^{(k)}$ has previously been considered it will have an associated path $p^{(k)}$, whose final element we will require to construct $E_{j}$: }to $E_{j-1}$ add the edges $(p^{(k)}_m,v^{(j)}_1)$, $\{(v^{(j)}_x,v^{(j)}_{x+1})\}_{x=1}^{l^{(j)}-m-1}$ and $(v^{(j)}_{l^{(j)}-m},N')$. This creates $E_j$. For $p^{(j)}$ we let $p^{(j)}_{1:m}=p^{(k)}_{1:m}$ and then add the new vertices of $V'_j$, and terminate the path at $N'$. If $l^{(j)}=m$ (and so the walk $\walk^{(j)}$ agrees with $\walk^{(k)}$ until its penultimate element) let $V'_{j}=V'_{j-1}$ and add the edge $(p^{(k)}_m,N')$ to $E'_{j-1}$. Then $p^{(j)}$ is equal to $p^{(k)}_{1:m}$ concatenated with $N'$. Finally our constructed graph is $G'=(V',E')=(V'_{\mathcal{M}},E'_{\mathcal{M}})$ with $N'=|V'|$.

Identifying $g(\walk^{(j)})=p^{(j)}$, we see that if $\walk^{(j)},\, \walk^{(k)}$ agree until their $m$th element then so also will $g(\walk^{(j)}),\, g(\walk^{(k)})$, but after this the paths diverge. We further require that the birth and death rates for the population type represented by the $k$th element of $\walk$ should be the same as the $k$th element of $g(\walk)$. Similarly for the transition rates between the $k$th and $k+1$th element of $\walk$ and $g(\walk)$.

\end{proof}
Coupling the above proposition with the results of Section \ref{sec_SIacyclic} allows one to characterise
$$
\min_{\walk\in \mathcal{W}_{1,N}^{(k)}} T(\walk)
$$
in terms of the weights of the walks. Further for a set of walks $\mathcal{W}_{1,N}^{(i)}$, \mic{Theorem \ref{thm_paths2}} holds but now with the walk weights, defined analogously to the definition given in \eqref{defSI_pathweight}. \mic{Note that the Proposition \ref{prop_acyctocyc} holds regardless of whether the root is a source vertex, or whether the vertex 1 population is the most fit. However, if transitions are permitted into the root vertex, then for any given walk, there is no guarantee that the cells associated with the first element of that walk are the most fit. This would leave us unable to apply the results derived in earlier sections.}

While the above proposition allows us to consider back transitions, the following result demonstrates that when the transition rates between vertices are small, the first initiation of the target is dominated by the paths, as opposed to walks containing cycles. This offers a secondary justification for mainly focusing on acyclic graphs. To state the result let us define $q:\mathcal{W}_{1,N}\mapsto\mathcal{P}_{1,N}$ as the operation reducing walks to their respective paths. This may be accomplished recursively by searching for the first cycle in $\walk$ and replacing this with the first element in the cycle. Then repeating the same procedure on the reduced walk until it is a path. This operation is the same as that which maps trajectories to their corresponding trajectory class in \cite{Tadrowski:2018}. Further, for $p\in \mathcal{P}_{1,N}$ let
$$
\mathcal{W}^{(i)}[p] = \{\walk\in \mathcal{W}_{1,N}^{(i)}: q[\walk]=p,\walk\neq p\}
$$
that is the walks of length at most $i-1$ which when reduced to paths are $p$, but excluding $p$ itself. Then with
$$
\nu_{\max}=\max\{ \nu(i,j): (i,j)\in E\}
$$
we have the following.

\begin{prop}\label{prop_nowalks}
For $i\geq 1$, $p\in \mathcal{P}_{1,N}$ such that $\mathcal{W}^{(i)}[p]$ is non-empty, and all $\walk\in \mathcal{W}^{(i)}[p]$
$$
\lim_{\nu_{\max}\rightarrow 0}\lim_{\nu^*\rightarrow 0} \pr(T(p)< T(\walk)|S_1)=1.
$$
\end{prop}
\begin{proof}
We construct an acyclic graph with appropriate vertex and edge parameters as in Proposition \ref{prop_acyctocyc} and use the mapping $g$ from there with domain $\mathcal{W}^{(i)}_{1,N}$. As before let $\mathcal{P}'_{1,|V'|} $ be the codomain of $g$. Then for path $p\in\mathcal{W}^{(i)}_{1,N}$ and walk $\walk\in \mathcal{W}^{(i)}[p]$, $g(p),\, g(\walk)\in \mathcal{P}'_{1,|V'|}$. Note due to the manner in which the edge parameters are chosen $\{\nu(i,N):\,i\in \mathcal{N}^{-}(N)\}=\{\nu(i,|V'|):\,i\in \mathcal{N}^{-}(|V'|)\}$ and hence $\nu^*$ remains unchanged. Similarly $\nu_{\max}$ is unchanged. Therefore,
\begin{equation}
\lim_{\nu^*\rightarrow 0} \pr(T(p)< T(\walk)|S_1)=\lim_{\nu^*\rightarrow 0} \pr(T'(g(p))< T'(g(\walk))|S_1).
\end{equation}	
Furthermore each vertex and edge in $p$ is also present in $\walk$. Hence the weight of $g(\walk)$ satisfies $w(g(\walk))=w(g(p))\tilde w(g(\walk))$ where $\tilde w(g(\walk))$ is the product of transition rates along the edges existing in cycles in $\walk$ divided by the costs of vertices in the cycles. Then by Theorem \ref{thm_paths}
\begin{equation}
\lim_{\nu^*\rightarrow 0} \pr(T'(g(p))< T'(g(\walk))|S_1)=\frac{1}{1+\tilde w(g(\walk))}.
\end{equation}
As $\tilde w(g(\walk)) $ contains transition rates, this leads to the stated result upon taking $\nu_{\max}\rightarrow 0$.
\end{proof}


\begin{redsect}

\subsection{Simulation for alternative transition formulation}
See Fig.\ \ref{fig_alt_trans} for a comparison of simulations with transitions following $(x)\rightarrow (y)$ with the hitting time results derived in Section \ref{sec_results} of the main text.

\begin{figure}[t]
	\centering
	\includegraphics[width=.6\linewidth,height=6cm]{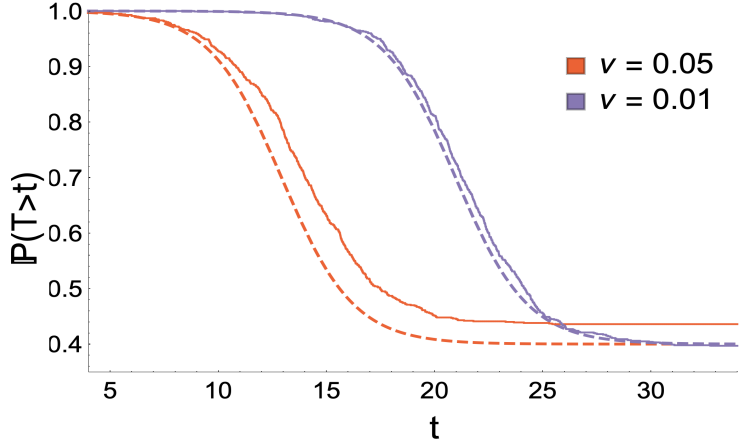}
	\caption{ \mic{Alternative transition formulation: Here we compare analogous simulations to Fig \ref{fig_growthtime}(b), except that transitions follow $(x) \rightarrow (y)$, with the theoretical approximation of \eqref{eqn_timeapprox}. Dashed lines give the theoretical result, joined lines represents the empirical distribution obtained from 1000 realisations. Despite \eqref{eqn_timeapprox} begin derived under the formulation when transitions follow $(x)\rightarrow (x),(y)$, there is excellent agreement between the simulations and theory. Parameters: $\alpha=\alpha(2)=\alpha(3)=\alpha(4)=1,\beta=0.4,\beta(2)=\beta(3)=1.3,\beta(4)=0.6, z=1$; legend displays value of $\nu=\nu(1,2)=\nu(2,3)=\nu(3,4)$ used.}}
	\label{fig_alt_trans}
\end{figure}



\subsection{Derivation of select formulas in Section \ref{sec_twodrugs}}\label{sec_deriv_twodrug}
Here we outline the derivation of the expressions contained in Section \ref{sec_twodrugs}, in particular formulas \eqref{cond_decreasetime}\eqref{eqn_mintime}\eqref{eqn_twodrug_prob}. These expressions were derived with the assistance of the Mathemetica software \cite{Mathematica}.

We start with \eqref{cond_decreasetime}. Let the set of stepwise paths be $\mathcal{P}_{\mathrm{SE}}$ and the set of directed paths be $\mathcal{P}_{\mathrm{DE}}$. For a path $p=(p_1,\ldots p_{l+1})$ with number of edges (path length) $l$, write  $c(p)=\prod_{i=2}^{l}(\growthone - \lambda(p_i))$. Note for any path considered here, if we follow the path on Fig.\ \ref{fig_twodrugstates}, then $c(p)$ is the product of the vertex labels, taken over all intermediate vertices (excluding the root and the target). Further, let us introduce $r=n_\rmd/n_\rms$ and $f = (n_\rmd+n_\rms)/n_{\mathrm{Tot}}$. This implies $n_\rmdd = (1-f)n_\rmt,\,n_\rms=f n_\rmt /(1+r),\, n_\rmd = r f n_\rmt /(1+r)$. Then 
\begin{align}\label{eqn_phitwodrug}
\phi_{N} &= \frac{n_\rmd}{n_\rmd+n_\rmdd}\frac{n_\rmdd}{n_\rms+n_\rmdd}m^2\nu^2\sum_{p\in \mathcal{P}_\mathrm{SE}} (c(p))^{-1}+\frac{n_\rmdd}{n_\rmdd+n_\rmd}\nu^2 m \sum_{p\in \mathcal{P}_\mathrm{DE}} (c(p))^{-1}
\\
&=\left[1+\frac{(1-f)(1+r)}{rf}\right]^{-1}\left[1+\frac{f}{(1+r)(1-f)}\right]^{-1} m^2\nu^2\sum_{p\in \mathcal{P}_\mathrm{SE}} (c(p))^{-1}\,+
\\ 
&\hspace{4mm}\left[1+\frac{r }{(1+r)(1-f)}\right]^{-1} \nu^2 m \sum_{p\in \mathcal{P}_\mathrm{DE}} (c(p))^{-1}.
\end{align}
Using the approximate formula \eqref{eqn_medtime}, we have $t_{1/2}(r)=-\lambda^{-1}\log(\lambda^{-1} \phi_{N} z))$. Also recall $z=\gamma n_s= \gamma f n_\rmt /(1+r)$. Differentiating $t_{1/2}(r)$ and evaluating at $r=0$ leads to 
$$
\frac{d}{d r}t_{1/2}(0) = \frac{\frac{1}{1-f}-\frac{fm(d+s)^3}{ds(d^2+3ds+s^2)}}{\growthone}.
$$
Increasing the relative size of the single drug compartment speeds up resistance when $\frac{d}{d r}t_{1/2}(0)<0$ or equivalently when
$$
\frac{ds(d^2+3ds+s^2)}{1-f}>f m(d+s)^2.
$$
Expansions in $f$ and $s/d$ (recall $f\ll 1, s/d\ll 1$), keeping only leading order terms and using the definition of $f$ leads to \eqref{cond_decreasetime}.

Turning to \eqref{eqn_mintime}, we wish to find $r$ such that $\frac{d}{d r}t_{1/2}(r)=0$. The numerator of $\frac{d}{d r}t_{1/2}(r)$ is quadratic in $r$. After solving for both roots and expanding in powers of $s/d$, we see the leading term for one of the roots is -1. Thus the root of interest is the other root, which to leading order is the expression in \eqref{eqn_mintime}.

For \eqref{eqn_twodrug_prob}, we use our main result \eqref{eqn_pathsthm}. With $\phi_N$ as given above in \eqref{eqn_phitwodrug}, this yields
\begin{align}
\pr(T^{\rm{SE}}<T^{\rm{DE}})& \approx \frac{ \frac{n_\rmd}{n_\rmd+n_\rmdd}\frac{n_\rmdd}{n_\rms+n_\rmdd} m^2\nu^2\sum_{p\in \mathcal{P}_\mathrm{SE}} (c(p))^{-1}}{\phi_N}
\\
& = \left( 1+\frac{(n_\rmdd+n_\rms)}{n_\rmd m} \frac{d^3 s +3 s^2 d^2 +d s^3}{(d+s)^3}\right)^{-1}.
\end{align}
Upon expanding in $s/d$, to leading order we have \eqref{eqn_twodrug_prob}.

\end{redsect}


\bibliographystyle{vancouver}
\bibliography{cancer4}

\end{document}